%% file: arxiv_journal_version.tex
\documentclass[12pt]{article}

\usepackage[T1]{fontenc}
\usepackage{amsmath}
\usepackage{amsfonts}
\usepackage{amssymb}
\usepackage{amsthm}
\usepackage{stmaryrd}
\usepackage{graphicx}
\usepackage{tikz}
\usepackage{esvect}
\usepackage[ruled,linesnumbered,vlined]{algorithm2e}
\usepackage{caption}
\usepackage{subcaption}
\usepackage{authblk}
\usepackage{mathtools}
\usepackage{cleveref}
\usepackage{todonotes}
\usepackage{comment}
\usepackage{thm-restate}
\usepackage{booktabs}
\usepackage{tabularray}
\usepackage[bibliography=common, appendix=strip]{apxproof}
\usepackage{lineno}
\usepackage[symbol]{footmisc}
\usepackage{fullpage}

\definecolor{malachite}{rgb}{0.04, 0.85, 0.32}

\usetikzlibrary{arrows,calc,shapes.misc}
\tikzstyle{noeud}=[scale=0.8,draw,circle,fill=black]
\tikzstyle{blocked}=[line width=2.8pt,red]
\tikzstyle{free}=[line width=1.4pt,black]
\newcommand{\convexpath}[2]{
	[   
	create hullnodes/.code={
		\global\edef\namelist{#1}
		\foreach [count=\counter] \nodename in \namelist {
			\global\edef\numberofnodes{\counter}
			\node at (\nodename) [draw=none,name=hullnode\counter] {};
		}
		\node at (hullnode\numberofnodes) [name=hullnode0,draw=none] {};
		\pgfmathtruncatemacro\lastnumber{\numberofnodes+1}
		\node at (hullnode1) [name=hullnode\lastnumber,draw=none] {};
	},
	create hullnodes
	]
	($(hullnode1)!#2!-90:(hullnode0)$)
	\foreach [
	evaluate=\currentnode as \previousnode using \currentnode-1,
	evaluate=\currentnode as \nextnode using \currentnode+1
	] \currentnode in {1,...,\numberofnodes} {
		-- ($(hullnode\currentnode)!#2!-90:(hullnode\previousnode)$)
		let \p1 = ($(hullnode\currentnode)!#2!-90:(hullnode\previousnode) - (hullnode\currentnode)$),
		\n1 = {atan2(\y1,\x1)},
		\p2 = ($(hullnode\currentnode)!#2!90:(hullnode\nextnode) - (hullnode\currentnode)$),
		\n2 = {atan2(\y2,\x2)},
		\n{delta} = {-Mod(\n1-\n2,360)}
		in 
		{arc [start angle=\n1, delta angle=\n{delta}, radius=#2]}
	}
	-- cycle
}

\makeatletter
\newskip\@bigflushglue \@bigflushglue = -100pt plus 1fil

\def\bigcentering{\let\\\@centercr\rightskip\@bigflushglue%
	\leftskip\@bigflushglue
	\parindent\z@\parfillskip\z@skip}

\makeatother

\DeclareMathOperator{\atan}{atan}
\DeclareMathOperator{\sh}{Sh}

\newcommand{\haut}{upper side\xspace}
\newcommand{\bas}{lower side\xspace}
\newcommand{\hcorde}{horizontal chord\xspace}
\newcommand{\vcorde}{vertical chord\xspace}
\newcommand{\hcordes}{horizontal chords\xspace}
\newcommand{\vcordes}{vertical chords\xspace}

\newcommand{\longversion}[1]{}

\newtheorem{theorem}{Theorem}
\newtheorem{definition}[theorem]{Definition}
\newtheorem{proposition}[theorem]{Proposition}
\newtheorem{lemma}[theorem]{Lemma}

\newtheorem{corollary}[theorem]{Corollary}

\newcommand{\kctp}{$k$-CTP\xspace}

\newcommand{\mcalf}{\mathcal{F}}

\newcommand{\card}[1]{\left| #1 \right|}

\newcommand{\dopt}{d_{\mbox{\scriptsize{opt}}}}

\newcommand{\popt}{P_{\mbox{\scriptsize{opt}}}}

\newcommand{\dopte}{d_{\emph{\scriptsize{opt}}}}

\newcommand{\dtr}[1]{d_{ #1 }^{\mbox{\scriptsize{Tr}}}}

\newcommand{\expo}{\textsc{ExpBalancing}\xspace}

\newcommand{\ve}[1]{\mbox{\bf #1}}

\title{The Canadian Traveller Problem on unit-weighted and  arbitrarily weighted outerplanar graphs\thanks{This work was supported by the International Research Center "Innovation Transportation and Production Systems" of the I-SITE CAP 20-25 and by the ANR project GRALMECO (ANR-21-CE48-0004).}}

\newcommand{\uca}{1}
\newcommand{\ulm}{2}
\newcommand{\inrae}{3}
\newcommand{\gscop}{4}

\author[\uca]{Laurent Beaudou}
\author[\uca]{Pierre Berg\'e}
\author[\uca,\ulm]{Vsevolod Chernyshev}
\author[\uca,\inrae]{Antoine Dailly}
\author[\uca]{Yan Gerard}
\author[\gscop]{Aurélie Lagoutte}
\author[\uca]{Vincent Limouzy}
\author[\uca]{Lucas Pastor}

\affil[\uca]{Universit\'e Clermont-Auvergne, CNRS, Mines de Saint-\'Etienne, Clermont-Auvergne-INP, LIMOS,
             France}
             
\affil[\ulm]{Ulm University,Germany}

\affil[\inrae]{Universit\'{e} Clermont Auvergne, INRAE, UR TSCF,
             France}
             
\affil[\gscop]{Univ. Grenoble Alpes, CNRS, Grenoble INP, G-SCOP, France}

\date{}

\begin{document}

\maketitle

\begin{abstract}
    We study the $k$-Canadian Traveller Problem, where a weighted graph $G=(V,E,\omega)$ with a source $s\in V$ and a target $t\in V$ are given. This problem also has a hidden input $E_* \subsetneq E$ of cardinality at most $k$ representing blocked edges. The objective is to travel from $s$ to $t$ with the minimum distance. At the beginning of the walk, the blockages $E_*$ are unknown: the traveller discovers that an edge is blocked when visiting one of its endpoints. Online algorithms, also called strategies, have been proposed for this problem and assessed with the competitive ratio, {\em i.e.}, the ratio between the distance actually traversed by the traveller divided by the distance he would have traversed knowing the blockages in advance.
	
	Even though the optimal competitive ratio is $2k+1$ even on unit-weighted planar graphs of treewidth 2, we design a polynomial-time strategy achieving competitive ratio 9 on unit-weighted outerplanar graphs. This value 9 also stands as a lower bound for this family of graphs as we prove that, for any $\varepsilon > 0$, no strategy can achieve a competitive ratio $9-\varepsilon$ on it. This comes actually from a strong connexion with another well-known online problem called the cow-path problem.

	Finally, we show that it is not possible to achieve a competitive ratio $e^{W(\frac{\ln k}{2})} - 1$ on arbitrarily weighted outerplanar graphs, where $W$ is the Lambert W function. This lower bound is asymptotically greater than $\frac{\ln k}{\ln \ln k}$.
\end{abstract}

\textit{Keywords : } Canadian Traveller Problem ; Online algorithms ; Competitive analysis ; Outerplanar graphs

\section{Introduction} \label{sec:intro}

The $k$-\textit{Canadian Traveller Problem} (\kctp) was introduced by Papadimitriou and Yannakakis~\cite{PaYa91}. It models the travel through a graph where some obstacles may appear. Given an undirected weighted graph $G=\left(V,E,\omega\right)$, with $\omega : E \rightarrow \mathbb{Q}^+$, and two of its vertices $s,t \in V$, a traveller walks from $s$ to $t$ on $G$ despite the existence of blocked edges $E_* \subsetneq E$ (also called \textit{blockages}), trying to contain the length of his walk. The traveller does not know which edges are blocked when he begins his journey. He discovers that an edge $e=uv$ is blocked, {\em i.e.}, belongs to $E_*$, when he visits one of its endpoints $u$ or $v$. The parameter $k$ is an upper bound on the number of blocked edges: $\card{E_*} \le k$.  Several variants have also been studied: where edges are blocked with a certain probability~\cite{AkSaAr16,BaSc91,BnFeSh09,fried2013complexity}, with multiple travellers~\cite{BeDeGuLe19,ShSa17}, where we can pay to sense remote edges~\cite{fried2013complexity}, or where we seek the shortest tour~\cite{HaXe23,LiHu14}. This problem has applications in robot routing for various kinds of logistics~\cite{AkSaAr16,AlYiAk21,BeBa23,EyKeHe10,LiScTh01}.

For a given walk on the graph, its \emph{cost} (also called distance) is the sum of the weights of the traversed edges. The objective is to minimize the cost of the walk used by the traveller to go from $s$ to $t$.
A pair $\left(G,E_*\right)$ is called a \textit{road map}.
All the road maps considered are feasible: there exists an $(s,t)$-path in $G\setminus E_*$, the graph $G$ deprived of $E_*$. In other words, there is always a way to reach target $t$ from source $s$ despite the blockages.

A solution to the \kctp is an online algorithm, called a \textit{strategy}, which guides the traveller through his walk on the graph : given the input graph, the history of visited nodes, and the information collected so far (here, the set of discovered blocked edges), it tells which neighbor of the current vertex the traveller should visit next. The quality of the strategy can be assessed with competitive analysis~\cite{BoEl98}. Roughly speaking, the \emph{competitive ratio} is the quotient between the distance actually traversed by the traveller and the distance he would have traversed knowing which edges are blocked in advance. The \kctp is PSPACE-complete~\cite{PaYa91,BaSc91} in its decision version that asks, given a positive number $r$ and the input weighted graph, whether there exists a strategy with competitive ratio at most $r$.
Westphal~\cite{We08} proved that no deterministic strategy achieves a competitive ratio less than $2k+1$ on all road maps satisfying $\card{E_*} \le k$.
Said differently, for any deterministic strategy $A$, there is at least one \kctp road map for which the competitive ratio of $A$ is at least $2k+1$.
Randomized strategies have also been studied, see \emph{e.g.}~\cite{BeWe15,DeHuLiSa14}.

\longversion{Randomized strategies, {\em i.e.}, strategies in which choices of directions depend on a random draw, have also been studied. Westphal~\cite{We08} proved that there is no randomized strategy achieving a ratio lower than $k+1$. 
Bender \emph{et al.}~\cite{BeWe15} studied graphs composed only of vertex-disjoint $(s,t)$-paths and proposed a polynomial-time strategy of ratio $k+1$. A slight revision of this strategy is reported in~\cite{ShSa19}.
Demaine {\em et al.} proposed a pseudo-polynomial-time randomized strategy on general graphs which achieves a competitive ratio $(1+\frac{\sqrt{2}}{2})k + O(1)$~\cite{DeHuLiSa14}.}

Our goal is to distinguish between graph classes on which the \kctp has competitive ratio $2k+1$ (the optimal ratio for general graphs) and the ones for which this bound can be improved. This direction of research has already been explored in~\cite{BeSa23}: there is a polynomial-time deterministic strategy which achieves ratio $\sqrt{2}k + O(1)$ on graphs with bounded-size maximum $(s,t)$-cuts. We pursue this study by focusing on a well-known family of graphs: outerplanar graphs, which are graphs admitting a planar embedding (without edge-crossing) where all the vertices lie on the outer face. In~\cite{BeSa23}, an outcome dedicated to a superclass of weighted outerplanar graphs implies that there is a strategy with ratio $2^{\frac{3}{4}}k + O(1)$ on them. 
Interestingly, however, even very simple unit-weighted planar graphs of treewidth $2$, consisting only of disjoint $(s,t)$-paths, admit the general ratio $2k+1$ as optimal~\cite{We08,ChChWuWu15}.

\medskip
\noindent\textbf{Our results and outline.} After some preliminaries (\Cref{sec:prelim}), we describe in \Cref{sec-unitweight} a polynomial-time strategy achieving a competitive ratio $9$ on instances where the input graph is a unit-weighted outerplanar graph:

\begin{restatable}{theorem}{unweightedOuterplanarRatio}
	\label{thm-unweightedOuterplanar-ratio9}
	There is a strategy with competitive ratio~9 for unit-weighted outerplanar graphs.
\end{restatable}

In the input outerplanar graph, vertices $s$ and $t$ lie on the outer face. The latter can be seen (provided 2-connectedness) as a cycle embedded in the plane, allowing to explore two sides when we travel from $s$ to $t$ (the two sides are the two internally disjoint $(s,t)$-paths forming the cycle). 
The core of the strategy consists in an exploration of both sides via a so-called \textit{exponential balancing}. Then, the most technical part consists in the handling of the chords linking both sides. We maintain a competitiveness invariant of the strategy which produces a final ratio of $9$.

Note that \Cref{thm-unweightedOuterplanar-ratio9} can be extended as a corollary to outerplanar graphs where the \emph{stretch}, defined as the ratio between the maximum and minimum weight, is bounded by some fixed $S$. In this case, the strategy has ratio $9S$.

Surprisingly, the \kctp on unit-weighted outerplanar graphs has connections with another online problem called the \textit{linear search} problem~\cite{BaCuRa93,BeNe70,Be63} or the \textit{cow-path} problem~\cite{KaReTa96}. In this problem, a traveller walks on an infinite line, starting at some arbitrary point, and its goal is to reach some target fixed by the adversary. It was shown that applying an exponential balancing on this problem is the optimal way, from the worst case point of view, to reach the target~\cite{BaCuRa93}. We prove in \Cref{subsec-lowerBound} that, on unit-weighted outerplanar graphs, the competitive ratio stated in \Cref{thm-unweightedOuterplanar-ratio9} is optimal, and we sketch how it could have been deduced from the literature on the linear search problem.

\begin{restatable}{theorem}{lowerbound}
    \label{thm:lower_bound}
    For any $\varepsilon > 0$, no deterministic strategy achieves competitive ratio $9-\varepsilon$ on all road maps $(G,E_*)$, where $G$ is a unit-weighted outerplanar graph.
\end{restatable}

Finally, in Section~\ref{sec-weighted}, we show that such constant competitive ratio cannot be achieved on arbitrarily weighted outerplanar graphs. Concretely, we identify a lower bound of competitiveness which is logarithmic in $k$ and depends on the Lambert W function~\cite{CoGo96}.

\begin{restatable}{theorem}{weightedBound}
    \label{thm:weightedbound}
    For any integer $k \ge 1$, no deterministic strategy can achieve a competitive ratio strictly less than $g(k) = e^{W(\frac{\ln k}{2})} - 1$ for the \kctp\ on weighted outerplanar graphs, where $W$ is the Lambert W function. Asymptotically, for sufficiently large $k$, $\frac{\ln k}{\ln \ln k} \le g(k) \le \ln k$.
\end{restatable}

We summarize in \Cref{tbl:summary} the state-of-the-art of the competitive analysis of deterministic strategies for the \kctp, giving for each family an upper bound of competitiveness (\emph{i.e.}, a strategy with such ratio exists) and a lower bound (\emph{i.e.}, no strategy can achieve a smaller ratio)\footnote{In the table, $\Omega$ is the Landau notation, converse of $O$.}. Our contributions are framed. A part of this article was already presented as an extended abstract in MFCS'24~\cite{BeBe24}. However, Theorem~\ref{thm:weightedbound} (proof in Section~\ref{sec-weighted}) is a new result.

\begin{table}[t]
  \centering
\scalebox{0.92}{
\begin{tblr}{lcc}
  \hline
  \SetCell{c} Family of graphs &  upper bound & lower bound \\
  \hline
   unit-weighted planar of treewidth 2  &  $2k+1$~\cite{We08} & $2k+1$~\cite{ChChWuWu15,We08} \\
   bounded maximum edge $(s,t)$-cuts & $\sqrt{2}k + O(1)$~\cite{BeSa23}  & ?\\
  outerplanar & $2^{\frac{3}{4}}k + O(1)$~\cite{BeSa23}  & \framebox{$e^{W(\frac{\ln k}{2})} - 1 = \Omega(\frac{\ln k}{\ln \ln k})$} \\
  unit-weighted outerplanar & \framebox{9} & \framebox{9}\\
  \hline
\end{tblr}
}
\caption{Deterministic strategies performances for the \kctp.}
\label{tbl:summary}
\end{table}

\section{Definitions and first observations} \label{sec:prelim}

\subsection{Graph preliminaries}

We work on undirected connected weighted graphs $G=(V,E,\omega)$, where $\omega : E \rightarrow \mathbb{Q}^+$.
A graph is \textit{equal-weighted} (resp. \textit{unit-weighted}) if the value of $\omega(e)$ is the same (resp. $1$) for every edge $e\in E$. 
This article follows standard graph notations from~\cite{Di12}. We denote by $G\left[U\right]$ the subgraph of $G$ induced by $U \subseteq V$: $G\left[U\right] = \left(U,E\left[U\right],\omega_{|E[U]}\right)$; and by $G\setminus U$ the graph deprived of vertices in $U$: $G\setminus U = G\left[V\setminus U\right]$.
A \textit{simple $(u,v)$-path} is a sequence of pairwise different vertices between $u$ and $v$, while, in a $(u,v)$-\textit{walk}, vertices can be repeated.
The \emph{cost} (or \emph{traversed distance}) of a walk or a path is the sum of the weights of the edges it traverses.
A vertex $v$ is an \emph{articulation point} if $G \setminus \{v\}$ is not connected.

An $(s,t)$-\textit{separator} $X \subsetneq V \setminus \{s,t\}$ in graph $G$ is a set of vertices such that $s$ and $t$ are disconnected 
in graph $G\setminus X$. We denote by $R_G(s,X)$ (resp. $R_G(t,X)$) the \textit{source} (resp. \textit{target}) \textit{component} of separator $X$, which is a set made up of the vertices of $X$ together with all vertices reachable from $s$ (resp. $t$) in $G\setminus X$.

A graph is \emph{outerplanar} if it can be embedded in the plane in such a way that all vertices are on the outer face. 
An outerplanar graph is 2-connected if and only if the outer face forms a cycle. Given an embedding of a
2-connected outerplanar graph $G = (V, E)$ and two vertices $s$ and $t$, let $s \cdot p_1 \cdot p_2 \cdots p_h \cdot t \cdot q_1 \cdot q_2 \cdots q_\ell \cdot s$ be the cycle along the outer face of $G$ and let $S_1 = \{p_1, p_2, \dots, p_h\}$ and $S_2 = \{q_1, q_2, \dots, q_\ell\}$ with $V = \{s, t\} \cup S_1 \cup S_2$. 
We can slightly deform the embedding so that
$s$ and $t$ are aligned along the horizontal axis; since the outer face forms a cycle, we will refer to $S_1$ (resp. $S_2$) as the \textit{upper} (resp. \textit{lower}) \textit{side} of $G$.
A chord $xy$ of the cycle formed by the outer face is said to be \emph{$(s,t)$-vertical} (resp. \emph{$(s,t)$-horizontal}) if $x$ and $y$ belong to different sides (resp. to the same side), see \Cref{fig-outerplanar-exemple}. When $x=s$ or $y=t$, the chord is considered as $(s,t)$-horizontal and not $(s,t)$-vertical.
Any $(s,t)$-vertical chord (simply \textit{vertical chord} when the context is clear) is an $(s,t)$-separator.
Considering a set of \vcordes, we say that the \textit{rightmost} one has the minimal inclusion-wise target component. Due to planarity, the rightmost \vcorde is unique for any such set.

\begin{figure}[ht]
	\centering
	\scalebox{1}{\input{tikzjournal/outerplanar-exemple}}
	\caption{Example of an outerplanar graph: $p_{2}q_\ell$, $p_2q_{\ell-1}$, $p_3q_{\ell-2}$, and $p_{h-1}q_{4}$ are vertical chords and $q_1q_3$, $q_1q_4$ are horizontal chords.}
	
	\label{fig-outerplanar-exemple}
\end{figure}
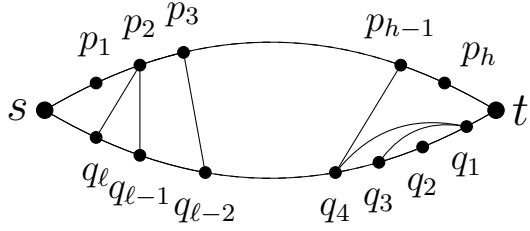

\subsection{Problem definition and competitive analysis}

Let $G=(V,E,\omega)$ be a graph and $E_*$ represent a set of blocked edges. A \textit{road map} is an instance for the \kctp . 

\begin{definition}[Road maps]
	A pair $(G,E_*)$ is a \textit{road map} if $s$ and $t$ are connected in $G\setminus E_*$.
	\label{def-roadmap}
\end{definition}
In other words, there must be an $(s,t)$-path in the graph $G$ deprived of the blocked edges $E_*$. 
We can now formally introduce the \kctp .

\begin{definition}[\kctp]~
	
	\textbf{\emph{Input:}} A graph $G=(V,E,\omega)$, two vertices $s,t \in V$, and a set $E_*$ of unknown blocked edges such that $|E_*|\leq k$ and $(G,E_*)$ is a road map.
	
	\textbf{\emph{Objective:}} Traverse graph $G$ from $s$ to $t$ with minimum cost.
\end{definition}

A solution to the \kctp is an $(s,t)$-walk.
The set of blocked edges $E_*$ is a hidden input at the beginning of the walk. We say an edge is \textit{revealed} when one of its endpoints has already been visited. A \textit{discovered blocked edge} is a revealed edge which is blocked.
At any moment of the walk, we usually denote by $E_*' \subseteq E_*$ the set of discovered blocked edges, in other words the set of blocked edges for which we visited at least one endpoint. 
During the walk, we are in fact working on $G\setminus E_*'$ as discovered blocked edges can be removed from $G$.

We call a path \textit{blocked} if one of its edges was discovered blocked;
\textit{apparently open} if no blocked edge has been discovered on it for now (it may contain a blocked edge which has not been discovered yet); \textit{open} if we are sure that it does not contain any blocked edge (either all of its edges were revealed open, or it is apparently open and $\vert E_*' \vert = k$, or by connectivity considerations since $s$ and $t$ must stay connected in road maps).

For any $F\subseteq E_*$ and two vertices $x,y$ of $G$, let $d_F\left(G,x,y\right)$ be the cost of the shortest $(x,y)$-path in graph $G\setminus F$. If the context is clear, we will use $d_F\left(x,y\right)$.

We denote by $\popt$ some \emph{optimal offline path} of road map $(G,E_*)$: it is one of the shortest $(s,t)$-paths in the graph $G\setminus E_*$.
Its cost, the \emph{optimal offline cost}, given by $\dopt = d_{E_*}\left(s,t\right)$, is the distance the traveller would have traversed if he had known the blockages in advance.
Given a strategy $A$ for the \kctp, the \textit{competitive ratio}~\cite{BoEl98} $c_A(G,E_*)$ over road map $(G,E_*)$ is defined as the ratio between the cost $\dtr{A}\left(G,E_*\right)$ of the traversed walk and $\dopt$.
Formally:
$$
	c_A(G,E_*) = \frac{\dtr{A}\left(G,E_*\right)}{\dopt}.
$$

Given a monotone family of graphs $\mcalf$ (i.e. closed under taking subgraph), we say that a strategy $A$ admits a competitive ratio $c(k)$ for the family $\mcalf$ if it is an upper bound for all values $c_A\left(G,E_*\right)$ over all \kctp\ road maps $(G,E_*)$ such that $G \in \mcalf$. Conversely, we say that some ratio $c(k)$ cannot be achieved for family $\mcalf$ if, for every strategy $A$, there is a road map $(G,E_*)$ with $G \in \mcalf$ such that $c_A(G, E_*) > c(k)$.

\begin{figure}[t]
	\centering
	\scalebox{0.65}{\input{tikzjournal/westphal}}
	\caption{Westphal graph with $k=4$, as defined in~\cite{We08}}
	\label{fig:westphal}
\end{figure}
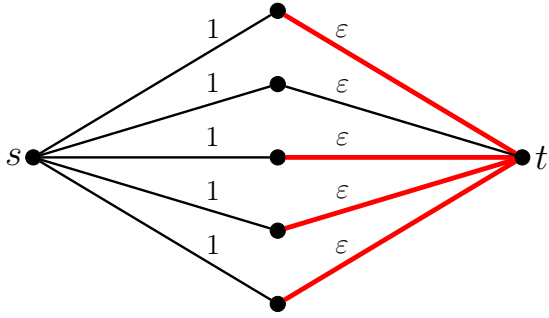

Westphal~\cite{We08} identified, for any integer $k$, a relatively trivial family of graphs for which any deterministic strategy achieves ratio at least $2k+1$ (see \Cref{fig:westphal}). These graphs are made up of only $k+1$ identical disjoint $(s,t)$-paths: they are planar and have treewidth $2$. As those paths are indistinguishable, the traveller might have to traverse $k$ of them before finding the open one. This outcome still works if we restrict ourselves to unit weights~\cite{ChChWuWu15}. 
Conversely, there are two strategies in the literature achieving competitive ratio $2k+1$ on general graphs: \textsc{reposition}~\cite{We08} and \textsc{comparison}~\cite{XuHuSuZh09}.

Note that articulation points allow a preliminary decomposition and simplification of any input graph, before even exploring:

\begin{lemma}
	\label{lem-articulationPoints}
    Let $\mcalf$ be a monotone family of graphs, and assume that we have a strategy $A$ achieving competitive ratio $C$ on graphs of $\mcalf$ that do not contain any articulation point. Then, there exists a strategy $A'$ achieving the same competitive ratio $C$ on all graphs of $\mcalf$. 
\end{lemma}
\begin{proof}
    The strategy $A'$ goes as follows: let $(G, E_*)$ be a road map with $G\in \mcalf$. If $G$ does not contain any articulation point, apply strategy $A$. Otherwise, let $z$ be an articulation point of $G$. If $\{z\}$ is not an $(s,t)$-separator, then, recursively apply strategy $A'$ on $R_G(s,\{z\})$, which is both the source and the target component, to reach $t$ from $s$. Otherwise (so $\{z\}$ is an $(s,t)$-separator), recursively apply strategy $A'$ on the source component $R_G(s,\{z\})$ to reach $z$ from $s$, then recursively apply strategy $A'$ on the target component $R_G(t,\{z\})$ to reach $t$ from $z$. The procedure is illustrated in \Cref{fig-unweightedOuterplanarGraphs-1-decomposition}. 
    
    \begin{figure}[h]
	\centering
	\scalebox{0.6}{\input{tikzjournal/unweighted-decomposition}}
	\caption{Decomposing the graph into components with no articulation points and removing the useless components (the vertices in a dashed rectangle are the same in the original graph). }
	\label{fig-unweightedOuterplanarGraphs-1-decomposition}
\end{figure}
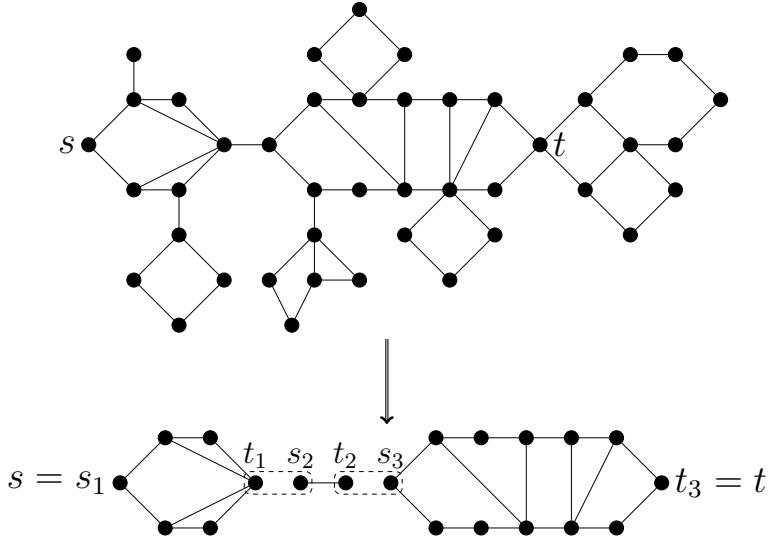
    
    We prove by induction on the number $p$ of articulation points that $A'$ terminates and achieves competitive ratio $C$. The base case $p=0$ holds by property of $A$. 
   
    For the inductive step, we distinguish two cases. If $\{z\}$ is not an $(s,t)$-separator, the walk we obtain is of length at most $C\dopt$, which gives competitive ratio $C$. Otherwise, the length of the whole walk at most $Cd_{E_*}(s,z)+Cd_{E_*}(z,t)$. Since $z$ is an $(s,t)$-separator, $z \in \popt$ and we have $\dopt = d_{E_*}(s,z) + d_{E_*}(z,t)$, which concludes the proof.
\end{proof}

\section{Optimal competitive ratio 9 for unit-weighted outerplanar graphs}
\label{sec-unitweight}

We propose a polynomial-time strategy called \expo dedicated to unit-weighted outerplanar graphs. We show that it achieves competitive ratio~9 for this family of graphs, which we will later prove is optimal (see \Cref{thm:lower_bound}). 

\subsection{Presentation of the strategy} \label{subsec-presentation}

First, note that \Cref{lem-articulationPoints} allows us to work on outerplanar graphs without articulation points. 
The input is a unit-weighted 2-connected outerplanar graph $G$ and two vertices $s$ and $t$. We provide a detailed description of the strategy \expo\ that we follow to explore the graph $G$.

\begin{enumerate}
	\item \label{step: reach t} \textbf{Reaching $t$}. If, at any point in our exploration, we reach $t$, then we exit the algorithm and return the processed walk.
	
	\item \label{step: reach horizontal chord} \textbf{Horizontal chords treatment}. If, at any point in our exploration, we visit a vertex $u \in S_i$, $i\in \{1,2\}$, incident with an open \hcorde $uv$ revealed for the first time, then we can remove all the vertices on side $S_i$ that lie between $u$ and $v$ on the outer face. Said differently, we get rid of the vertices which are surrounded by the chord $uv$. If several \hcordes incident with $u$ are open, then it suffices to apply this rule to the chord which surrounds all others. This procedure comes from the observation that, due to both unit weights and planarity, the open \hcorde $uv$ with the rightmost $v$ is necessarily the shortest way to go from $u$ to $t$ on side $S_i$, and thus visiting the vertices surrounded by it will occur an extra, useless cost. 
	
	\item \label{step: exponential strategy} \textbf{Exponential balancing}. The core \textit{exponential balancing} principle of the strategy consists in alternately exploring sides within a given budget that doubles each time we switch sides. The budget is initially set to 1. Hence, we walk first on side $S_1$ with budget $1$, second on side $S_2$ with budget $2$, then on side $S_1$ with budget $4$, and so on. We say each budget corresponds to an \textit{attempt}. During each attempt, we traverse a path starting from the source $s$ and stay exclusively on some side $S_i$, $i \in \{1,2\}$. As evoked in the previous step, at each newly visited vertex, we use an open \hcorde from our position which brings us as close as possible to $t$ on our side. Either a \hcorde is open and we use the one which surrounds all other open chords, or if no such chord is open, we pursue our walk on the outer face.

	This balancing process can be described on an automaton depicted in \Cref{fig-automateDeBeaudou} which will be particularly useful in the analysis of this strategy. Here, we assume that we neither are completely blocked on one side nor reveal an open \vcorde . We will handle these cases in Steps~\ref{step: blocked edge}-\ref{step: open chord between C and A}.
	
	 We start our walk on $s$ (state $\mathbf{E_1}$), make an attempt on an arbitrary side (say $S_1$) with budget $1$ (state $\mathbf{E_2}$), and decide to come back to $s$ if $t$ was not reached. During our first attempt on side $S_2$ with budget 2, we cross a first edge and reach state $\mathbf{A}$. Then, we cross a second edge if we are not blocked, but this part of the journey corresponds to the transition between states $\mathbf{A}$ and $\mathbf{B}$. The automaton works as follows:
	 
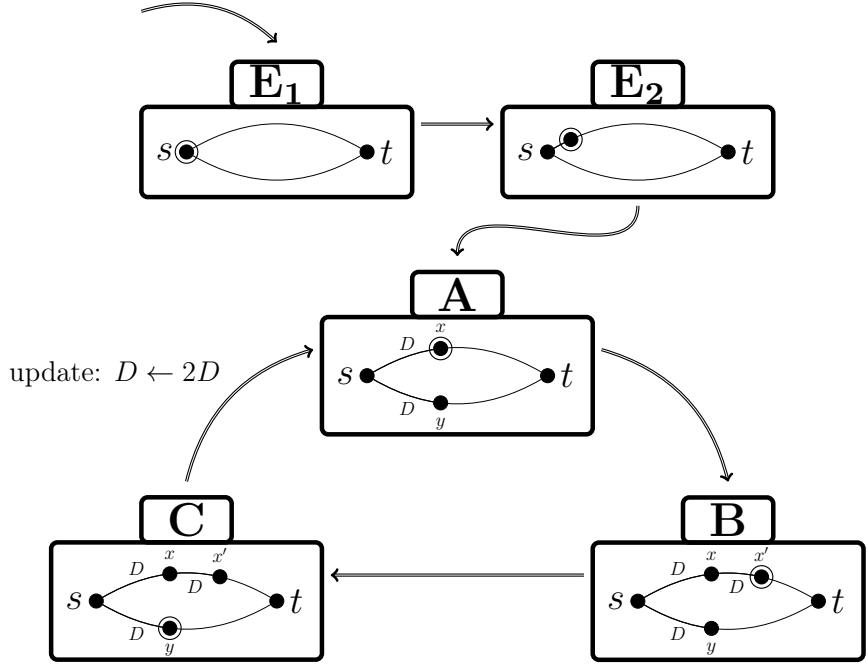
\begin{figure}[t]
	\centering
	\scalebox{0.6}{\input{tikzjournal/automateDeBeaudou}}
	\caption{Representation of the exponential balancing divided into three different states. The circled vertex is the one we are currently exploring.}
	\label{fig-automateDeBeaudou}
\end{figure}

	\begin{itemize}
		\item In state $\mathbf{A}$, we have explored $D$ vertices on each side (in the description above, $D=1$ when we first arrive in state $\mathbf{A}$). Call $x$ and $y$ the last explored vertices on each side, assume we are on $x$. The current budget is $2D$ and we pursue our attempt on the side of $x$.
		\item We then explore at most $D$ more vertices on the side of $x$. We reach state $\mathbf{B}$.
		\item We then go back to $y$ through $s$, reaching state $\mathbf{C}$.
		\item We explore at most $D$ more vertices on the side of $y$. We go back to state $\mathbf{A}$ with an updated value of $D$ that is doubled, update $x$ and $y$, and the sides are switched.
	\end{itemize}

	\item \label{step: blocked edge}  \textbf{Bypassing a blocked side}. If, during some attempt on side $S_i$, we are completely blocked (there is no open $(s,t)$-path on $G[S_i]\setminus E_*'$) before reaching the budget, hence exploring $\alpha D$ ($\alpha<1$) instead of $D$ (see \Cref{fig: blocked edge between A and B,fig: blocked edge between C and A}), then we backtrack to $s$ and pursue the balancing on the other side $S_j$ ($j \in \{1,2\}, j\neq i$). However, we forget any budget consideration: we travel until we either reach $t$ or visit the endpoint $u$ of some open \vcorde $uv$. In case there are several open \vcordes incident with $u$ revealed at the same time, we consider the rightmost one. At this moment, we update the current graph $G\setminus E_*'$ by keeping only the target component of separator $\{u,v\}$ and considering $u$ as a new source. Concretely, we concatenate the current walk computed before arriving at $u$ with a recursive call of \expo on input $(G[R_G(t,\{u,v\})],u,t)$.
	
	\item \label{step: open chord between A and B} \textbf{Handling open \vcordes between states $\mathbf{A}$ and $\mathbf{B}$}. If, during some attempt on side $S_i$, especially in the transition between states $\mathbf{A}$ and $\mathbf{B}$, we reveal an open \vcorde $uv$, $u \in S_i$, after having explored distance $\alpha D$ (parameter $\alpha$ is rational, $0 < \alpha \leq 1$, but $\alpha D$ is an integer), then we go to the other side $S_j$, $j\neq i$, through $uv$ and explore side $S_j$ from $v$ towards $s$ until we:
	\begin{itemize}
	    \item either ``see'' a vertex $y$ already visited after distance $\beta D$ (we fix $\beta D \le \alpha D - 1$, so $0 \leq \beta < \alpha$),
	    \item or explore distance $\alpha D - 1$ and do not see any already visited vertex,
	    \item or are completely blocked on $S_j$ before we reach distance $\alpha D - 1$. 
	\end{itemize} 
	By ``see'', we mean that we can reach - or not - a neighbor of $y$ which reveals the status of the edge between them: in this way, we actually know the distance to reach $y$ from $v$ even if we did not visit $v$. \Cref{fig: open chord between A and B} describes this rule with an example.
	
	The role of this procedure is to know which endpoint of the chord is the closest to $s$. If we see, after distance $\beta D = \alpha D - 1$, an already visited vertex (denoted by $y$ in \Cref{fig: open chord between A and B}) at distance $\alpha D$ from $v$, then, we continue the exponential balancing: we go back to $v$ and thus to state $\mathbf{A}$ in the automaton, update the budget value $D$ which becomes $D+\alpha D$.
	
	Otherwise, we update $G$ by keeping only the target component of separator $\{u,v\}$. The current graph becomes $G' = G[R(t,\{u,v\})]$. If we saw an already visited vertex $y \in S_j$ by exploring distance $\beta D < \alpha D - 1$, then the new source becomes $s' = v$. Otherwise, the new source is $s' = u$. We concatenate the current walk with the walk returned by applying \expo on input $(G',s',t)$.
	
\begin{figure}[t]
\centering
\input{tikzjournal/steps-strategy}

\caption{Four situations potentially met with \expo on some unit-weighted outerplanar graph. The circled vertex is the current position.}
\label{fig:strategy}
\end{figure}
	
	\item \label{step: open chord between C and A} \textbf{Handling open \vcordes between states $\mathbf{C}$ and $\mathbf{A}$}. If, during the transition between states $\mathbf{C}$ and $\mathbf{A}$ (when some attempt is launched on the side of $y$ and the traversed distance on the other side is larger, see \Cref{fig-automateDeBeaudou}), an open \vcorde $uv$ is revealed (see \Cref{fig: open chord between C and A}), then we keep only the target component of $\{u,v\}$ and set $u$ as the new source. More formally, we concatenate the current walk with the walk returned by applying \expo on input $(G',u,t)$, where $G' = G[R_G(t,\{u,v\})]$.
\end{enumerate}

Steps~\ref{step: blocked edge}-\ref{step: open chord between C and A} can be summarized in this way: when we reveal an open \vcorde $uv$ such that $d_{E_*}(s,v) = d_{E_*}(s,u) + 1$, we launch a recursive call on the target component of separator $\{u,v\}$ with source $u$ and target $t$. Indeed, any optimal offline path must pass through separator $\{u,v\}$ and, as $d_{E_*}(s,v) = d_{E_*}(s,u) + 1$, we can say there is one optimal offline path $\popt$ such that $u \in \popt$. Hence, it makes sense to select $u$ as a new source, there is no interest in visiting 
vertices different from $\{u,v\}$ belonging to their source component.

\subsection{Competitive analysis} \label{subsec-competitive}

We now show that the strategy \expo presented above has competitive ratio $9$ on unit-weighted outerplanar graphs. We prove this statement by minimal counterexample. In this subsection, let $G$ denote the smallest (by number of vertices, then number of edges) unit-weighted outerplanar graph on which \expo does not achieve competitive ratio~9. We will see that the existence of such a graph $G$ necessarily implies a contradiction.

Examples of executions of \expo are given in \Cref{fig-unweightedOuterplanarGraphs-2-firstGraph,fig-unweightedOuterplanarGraphs-3-secondGraph}.

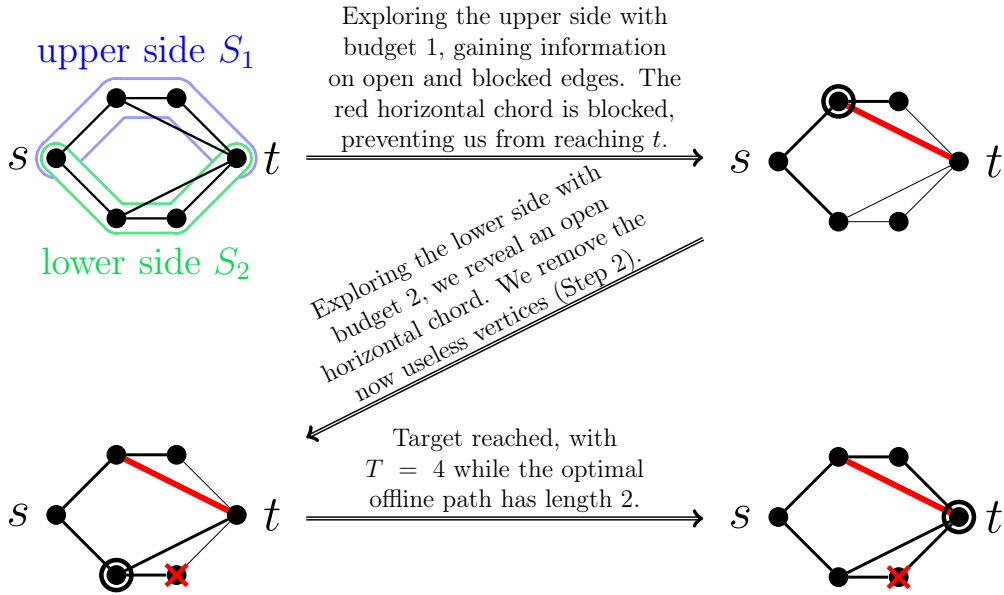
\begin{figure}[h]
	\centering
	\scalebox{0.8}{\input{tikzjournal/unweighted-first}}
	\caption{Application of \expo on the first graph of the decomposition of \Cref{fig-unweightedOuterplanarGraphs-1-decomposition}. At each step, the circled vertex is the one we are currently exploring, and we know the status of the bold edges: black is open, red is blocked.}
	\label{fig-unweightedOuterplanarGraphs-2-firstGraph}
\end{figure}

The two following technical lemmas prove that a recursive call has to happen when \expo is applied on $G$ (\Cref{lem-recursive}) and that such a recursive call implies certain properties (\Cref{lem-vcorde}).

\begin{lemma}
 During the execution of \expo, let $T$ be the  distance travelled at a given point before the first recursive call (if any). Then, $T\leq 9\dopte$. Moreover, if we are in state $\mathbf{A}$, let $x$ and $y$ be the last two vertices explored on each side during the exponential balancing. Then:
		$(i)$ $d_{E_*}\left(s,x\right)=D$, $(ii)$ $d_{E_*}\left(s,y\right)=D$ and $(iii)$ $T \leq 5D$.
  
\label{lem-recursive}
\end{lemma}
\begin{proof}
Assume that we have applied \expo on $G$ until a certain point and that no recursive call was launched so far.
We first focus on the second part of the invariant we want to show :
	\begin{center}
		In state $\mathbf{A}$, $(i)$ $d_{E_*}\left(s,x\right)=D$, $(ii)$ $d_{E_*}\left(s,y\right)=D$ and $(iii)$ $T \leq 5D$.
	\end{center}
	Items $(i)$ and $(ii)$ are true, since no shortcut between $s$ and either $x$ or $y$ can exist: any open \hcorde is used, and an open \vcorde opening up a shortcut leads to a recursive call (Steps~\ref{step: open chord between A and B} and~\ref{step: open chord between C and A}).
	
	Item $(iii)$ is trivially true when we kick-start the exponential balancing: when entering $\mathbf{A}$ from $\mathbf{E_2}$, we have $T=3$ and $d_{E_*}\left(s,x\right)=d_{E_*}\left(s,y\right)=1$. Assume that it is true for a given $D \geq 1$, and let $T_0$ be the value of $T$ at this point. When we reach state $\mathbf{B}$, we have $T = T_0+D \leq 6D$. When we reach state $\mathbf{C}$, we have $T = T_0+D+3D \leq 9D$. In brief, from state $\mathbf{A}$ to $\mathbf{C}$, we have $\dopt \ge D$ as distance $D$ was explored on both sides without reaching $t$. The largest ratio of $T$ by $D$ on these phases is 9 at state $\mathbf{C}$, where we have $T \le 9\dopt$.
	
	During the transition from $\mathbf{C}$ to $\mathbf{A}$, if $D + \alpha D$ denotes the traversed distance on current side at any moment (see \Cref{fig-automateDeBeaudou}), then $\dopt \ge D + \alpha D$ and $T = 9D + \alpha D$. The ratio $\frac{T}{d_{\mbox{\tiny{opt}}}}$ admits a decreasing upper bound, from 9 in state  $\mathbf{C}$ to 5 in $\mathbf{A}$. Indeed, when we are back to state $\mathbf{A}$, we have $T = T_0+D+3D+D$, but the value of $D$ is updated. Let $D' = 2D$. We have $T = T_0+5D \leq 5D+5D = 5D'$, and so item $(iii)$ remains true during the core loop. 
	
	We also have to check that it is true when we met an open \vcorde $uv$ between states $\mathbf{A}$ and $\mathbf{B}$ which satisfies $d_{E_*}(s,v) = d_{E_*}(s,u)$ (case $\beta D  = \alpha D - 1$ in Step~\ref{step: open chord between A and B}). In this case, the new value of $D$ is $D'=D+\alpha D$ and we have $T \leq 5D + \alpha D + 1 + 2\alpha D \leq 5(D+\alpha D) = 5D'$ (since $\alpha D \geq 1$), so item $(iii)$ remains true.

Thus, conditions $(i)$-$(iii)$ hold in state $\mathbf{A}$, and we always (during all states and transitions between them) have $T\leq 9\dopt$, hence the statement holds. 

\end{proof}

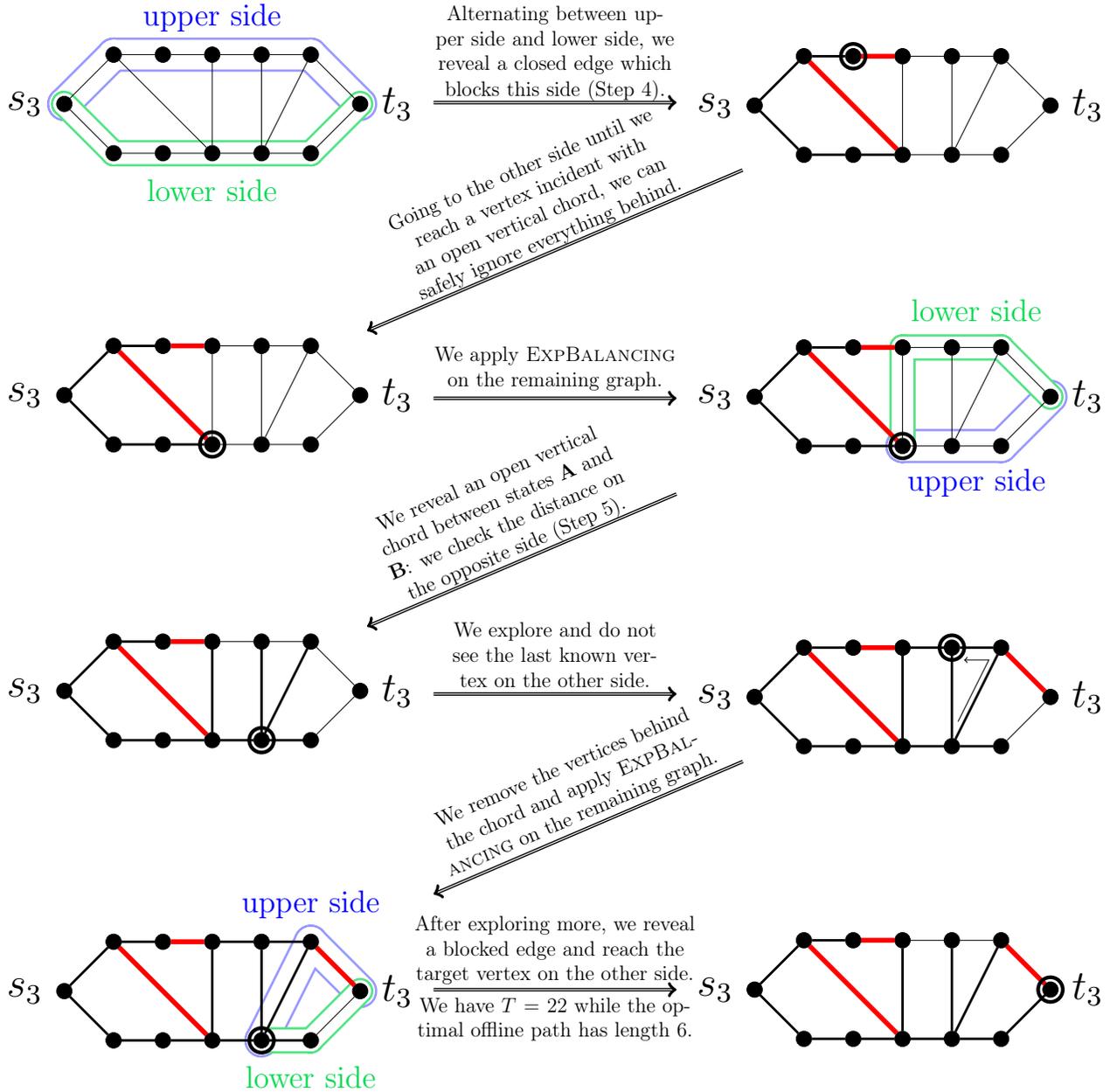
\begin{figure}
	\centering
	\begin{bigcenter}
		\scalebox{0.75}{\input{tikzjournal/unweighted-second}}
	\end{bigcenter}
	\caption{Application of \expo on the third graph of the decomposition of \Cref{fig-unweightedOuterplanarGraphs-1-decomposition}. At each step, the circled vertex is the one we are currently exploring, and we know the status of the bold edges: black is open, red is blocked.}
	\label{fig-unweightedOuterplanarGraphs-3-secondGraph}
\end{figure}

\begin{lemma}
Assume that we are currently executing \expo on $G$ and that a recursive call is launched after revealing the \vcorde $uv$ with new source $u$. Let $T$ be the distance traversed before the recursive call. Then, either $T > 9d_{E_*}\left(s,u\right)$ or $d_{E_*}\left(s,v\right) < d_{E_*}\left(s,u\right)+1$.
\label{lem-vcorde}
\end{lemma}
\begin{proof}
	If $d_{E_*}\left(s,v\right) \ge d_{E_*}\left(s,u\right)+1$, following the rules established in Steps~\ref{step: blocked edge}-\ref{step: open chord between C and A}, we will launch a recursive call on the target component of $\{u,v\}$ with new source $u$. Hence, we will have $\dtr{\textsc{exp}}\left(G,E_*\right)= T+T'$, where $T' \leq 9d_{E_*}\left(u,t\right)$ by minimality of $G$ and \textsc{exp} abbreviates \expo. By way of contradiction, suppose that $T \leq 9d_{E_*}\left(s,u\right)$. The optimal offline path $\popt$ necessarily goes through the separator $\{u,v\}$ in graph $G$ and, since $d_{E_*}\left(s,v\right)=d_{E_*}\left(s,u\right)+1$, $u$ belongs to some optimal offline path. Consequently, $T+T' \leq 9(d_{E_*}\left(s,u\right)+d_{E_*}\left(u,t\right)) = 9d_{E_*}\left(s,t\right)$.
	\end{proof}

We are now ready to prove the major contribution of this article.
	
\unweightedOuterplanarRatio*
\begin{proof}
    A direct consequence of \Cref{lem-recursive} is that, during some attempt, \expo will launch a recursive call on $G$ (otherwise, it has competitive ratio~9, a contradiction).
	Let $T$ be the distance traversed before the recursive call.
	\Cref{lem-vcorde} has an important consequence: if we launch a recursive call on the open \vcorde $uv$ with new source $u$ and can guarantee that both $d_{E_*}\left(s,v\right)=d_{E_*}\left(s,u\right)+1$ and $T \leq 9d_{E_*}\left(s,u\right)$, then, we have a contradiction. According to the description of \expo, a recursive call is launched when we are sure that $d_{E_*}\left(s,v\right)=d_{E_*}\left(s,u\right)+1$: this concerns Step~\ref{step: blocked edge}, Step~\ref{step: open chord between A and B} when $\beta D < \alpha D - 1$ and Step~\ref{step: open chord between C and A}.

	Assume first that we are blocked on one side between states $\mathbf{A}$ and $\mathbf{B}$ in Step \ref{step: blocked edge} (see \Cref{fig: blocked edge between A and B}). We know that $d_{E_*}\left(s,v\right)=d_{E_*}\left(s,u\right)+1$ because $u$ is an articulation point of $G\setminus E_*'$. Also, $d_{E_*}\left(s,u\right)=D+d_{E_*}\left(y,u\right)$. Using \Cref{lem-recursive}:
	$$
	\begin{array}{rclr}
		T & \leq & (5D + \alpha D) + (\alpha D + 2D) + d_{E_*}\left(y,u\right) & \\
		& \leq & (7 + 2\alpha)D + d_{E_*}\left(y,u\right) & \\
		& \leq & 9(D + d_{E_*}\left(y,u\right)) & (\alpha \leq 1) \\
		& \leq & 9d_{E_*}\left(s,u\right) &
	\end{array}
	$$
	which, by \Cref{lem-vcorde}, leads to a contradiction.

	Assume now that we are blocked on one side between states $\mathbf{C}$ and $\mathbf{A}$ in Step~\ref{step: blocked edge} (see \Cref{fig: blocked edge between C and A}). Let $x'$ be the last vertex reached at the end of state $\mathbf{A}$, we know that $d_{E_*}\left(s,v\right)=d_{E_*}\left(s,u\right)+1$ because $u$ is an articulation point of $G\backslash E_*'$ and $d_{E_*}\left(s,u\right)=2D+d_{E_*}\left(x',u\right)$. Using \Cref{lem-recursive}:
	$$
	\begin{array}{rclr}
		T & \leq & (9D + \alpha D) + (\alpha D + 3D) + d_{E_*}\left(x',u\right) & \\
		& \leq & (12 + 2\alpha)D + d_{E_*}\left(x',u\right) & \\
		& \leq & 9(2D + d_{E_*}\left(x',u\right)) \leq 9d_{E_*}\left(s,u\right) & ~~~(\alpha \leq 1) \\
	\end{array}
	$$
	which, by \Cref{lem-vcorde}, leads to a contradiction.
	
	Assume now that we reveal an open \vcorde $uv$ between states $\mathbf{A}$ and $\mathbf{B}$ in Step~\ref{step: open chord between A and B} (see \Cref{fig: open chord between A and B}). Recall that $d_{E_*}\left(s,u\right) \leq D+\alpha D$, and we explore up to distance $\alpha D - 1$ towards $y$. There are two possibilities: either we see $y$ by exploring distance $\beta D$ (with $\beta D < \alpha D - 1$), or we do not see $y$ even if we explore distance $\alpha D - 1$.
	
	If we see $y$, then, we know that $d_{E_*}\left(s,u\right)=d_{E_*}\left(s,v\right)+1$ since going to $u$ through $x$ will yield distance $D+\alpha D$ while going through $y$ and $v$ will yield distance at most $D + \beta D + 2$, and we know that $\beta D < \alpha D - 1$ and $\beta D \geq 0$. So, $d_{E_*}\left(s,v\right) = D + \beta D + 1$. Using \Cref{lem-recursive}:
	$$
	\begin{array}{rclr}
		T & \leq & (5D + \alpha D) + (1 + 2(\beta D + 1)) & \\
		& \leq & (5 + \alpha + 2\beta)D+3 & \\
		& \leq & 9(D + \beta D + 1) \leq 9d_{E_*}\left(s,v\right) & ~~~(\beta < \alpha \leq 1) \\
	\end{array}
	$$
	which, by \Cref{lem-vcorde} leads to a contradiction (the roles of $u$ and $v$ are reversed here, since $v$ is the new source).
	
	If we do not reach $y$, either by blocked edges or because we have explored distance $\alpha D - 1$ without reaching it, then, we know that $d_{E_*}\left(s,v\right)=d_{E_*}\left(s,u\right)+1$. Using \Cref{lem-recursive}:
	$$
	\begin{array}{rclr}
		T & \leq & (5D + \alpha D) + (1 + 2(\alpha D - 1) + 1) & \\
		& \leq & (5 + 3\alpha)D & \\
		& \leq & 9(D + \alpha D) \leq 9d_{E_*}\left(s,u\right) & ~~~(\alpha \leq 1) \\
	\end{array}
	$$
	which, by \Cref{lem-vcorde}, leads to a contradiction.
	
	Finally, assume that we reveal an open \vcorde $uv$ between states $\mathbf{C}$ and $\mathbf{A}$ after having explored $\alpha D$ vertices in Step~\ref{step: open chord between C and A} (see \Cref{fig: open chord between C and A}). Since $uv$ was not revealed before, this implies that the shortest path from $s$ to $v$ goes through $u$, and so $d_{E_*}\left(s,v\right)=d_{E_*}\left(s,u\right)+1$. Using \Cref{lem-recursive}:
	$$
	\begin{array}{rclr}
		T & \leq & 9D + \alpha D & \\
		& \leq & 9(D + \alpha D) \leq  9d_{E_*}\left(s,u\right) & ~~~(\alpha \leq 1) \\
	\end{array}
	$$
	which, by \Cref{lem-vcorde}, leads to a contradiction.
	
	All the possible cases lead to contradictions, and so such a $G$ cannot exist. \expo thus achieves competitive ratio $9$ on unit-weighted outerplanar graphs.
    \end{proof}

The strategy \expo\ can also be applied on outerplanar graphs where the stretch $S$ is bounded.

\begin{corollary}
There is a strategy with competitive ratio $9S$ on outerplanar graphs of stretch $S$.
\label{co-stretch}
\end{corollary}
\begin{proof}
Apply strategy \expo\ as if the graph was unit-weighted. Let $\alpha$ be the minimum weight of the input graph and $W_{\mbox{\scriptsize{opt}}}$ be the number of edges of the optimal offline path. The total distance traversed is upper-bounded by $9S\alpha W_{\mbox{\scriptsize{opt}}}$ while $\dopt \ge \alpha W_{\mbox{\scriptsize{opt}}}$.
\end{proof}

\subsection{Lower bound 9 for unit-weighted outerplanar graphs}
\label{subsec-lowerBound}

In this subsection we prove that the competitive ratio achieved with the \expo strategy is optimal on unit-weighted outerplanar graphs.

 \label{subsec-lowunit}

This result can be obtained by a natural reduction from the linear search problem \cite{Be63} (or, equivalently, the cow-path problem on two rays \cite{KaReTa96}). The \emph{linear search problem} is defined as follows:
an immobile hider is located on the real line. A searcher starts from the origin and wishes to discover the hider in minimal time. The searcher cannot see the hider until he actually reaches the point at which the hider is located and the time elapsed until this moment is the duration of the game.

This problem reduces to the \kctp on specific road maps that we call \emph{shell road maps}. The {\em shell graph} on $2n$ vertices, denoted by $\sh_n$ (see \Cref{fig:dailly_graph}), is the graph obtained from a cycle on $2n$ vertices $\{v_0,v_1,\ldots,v_{2n-1}\}$ with all possible chords incident with vertex $v_n$, except $v_0v_n$. It is clearly outerplanar, and all edge weights are set to 1. In our setting, we shall consider $v_0$ as the source $s$ and $v_n$ as the target $t$.  We call \emph{shell road maps} the specific road maps $(\sh_n,E_*)$ where $E_*$ is made up only of edges incident with $t$. Said differently, the traveller cannot be blocked on the outer face on some edge $v_iv_{i+1}$.  

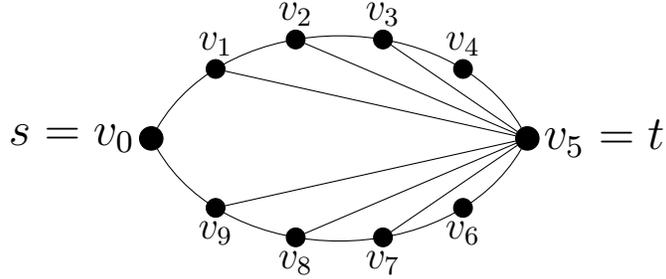
\begin{figure}[h]
    \centering
    \scalebox{1}{\input{tikzjournal/grapheDeDailly}}
    \caption{The shell graph on 10 vertices $\sh_5$.}
    \label{fig:dailly_graph}
\end{figure}

The shell graph is 2-connected, so it contains an upper side $S_1$ and a lower side $S_2$ which can simulate the positive and the negative sides of the real line. The position of the hider will then intuitively correspond to the first encountered open chord to $t$ : if the hider is at position $x>0$ (resp. $x<0$), then $E_*$ will contain all $v_it\in E$ except $v_{\lceil x \rceil}t$ (resp. $v_{\lfloor 2n+x \rfloor}t$). In such a way, any strategy for the \kctp with some competitive ratio $r$, will give a strategy for the linear search problem with asymptotic competitive ratio $r+\varepsilon$ for any $\varepsilon>0$. However, it is known that the linear search problem has an optimal ratio of 9 \cite{BaCuRa93} which gives the lower bound we want on the \kctp. Note however that, in the sketched reduction, small details need to be cared of, for example the distance of the traveller has a unit additive term compared to the searcher on the line (cost of crossing the discovered chord to $t$).
In order to remove any doubt related to these details, we provide below a complete proof of \Cref{thm:lower_bound} (without reducing to the linear search, but sharing some features with the proof of \cite{BaCuRa93}).

\textbf{Formal encoding of a strategy on shell graphs}. Given some positive integer $n$, let $A_n$ be an optimal strategy for graph $\sh_n$, {\em i.e.}, a strategy which minimizes the competitive ratio obtained on road maps $(\sh_n,E_*)$. A first observation is that, following $A_n$, when the traveller stands on some vertex $v_i$, he should always traverse a chord $v_it$ if it is open (the ratio can only increase if the traveller decides to explore the graph a bit longer). Furthermore, when traversing an already visited section of a side, $A_n$ should do it directly, with a simple walk, and avoid multiple crossings of the same edge. Also note that if the traveller, following $A_n$, starts switching sides by coming back to $s$, he should not change his mind and go back to exploring his side, since doing so would incur a cost for no additional information (no edge can be revealed this way). Finally, if the traveller has already explored $\ell$ vertices on a side, when he goes back to exploring this same side, he should always traverse at least $\ell+1$ vertices (so reveal the edges incident with at least one more vertex) before switching again. Doing otherwise would incur a cost for no additional information.

Thus strategy $A_n$ can be described as a sequence of integers $(x_i)_{i\ge 1}$ where $x_i$ represents the budget distance we afford ourselves after coming back to $s$ and switching side. Concretely, first, the traveller selects one side arbitrarily (say $S_1$) and traverses a distance of $x_1$ on the outer face. If an open chord is revealed during his walk, he reaches $t$. Otherwise, after being forced to stay on the outer face, he has no choice but backtrack towards $s$ and traverse the other side $S_2$ for some distance $x_2$. Generally, value $x_i$ denotes the distance budget we allow ourselves to traverse for the $i$th attempt, (on upper side $S_1$ for odd $i$ and bottom side $S_2$ for even $i$) before coming back to $s$ if no \hcorde to $t$ was found. The strategy ends whenever an open edge incident with $t$ is found and traversed. Observe that all values $x_i$ are at least 1. Furthermore, by the last observation of the previous paragraph, we always have $x_{i+2} > x_i$.

Assume the traveller reaches target $t$ on attempt $j+1$. Let $S_i$ be the last side visited, $i \in \{1,2\}$. Compatible with such travel hypothesis, we consider the road map $(\sh_n,E_*)$ with the following blocked edges in $E_*$.
\begin{itemize}
    \item all chords $ut$ where $u$ is on $S_i$ and is at distance at most $x_{j-1}$ from $s$,
    \item all chords $ut$ where $u$ is not on $S_i$.
\end{itemize}

In this way, we force the traveller to reach $t$ via the vertex which was placed just after the last one he visited during attempt $j-1$. We have $\dopt = x_{j-1} + 2$ and the total traversed distance from the beginning of the walk is $2(\sum_{i = 1}^{j}x_i) + x_{j-1} + 2$.

\textbf{Analysis}. Assume that $A_n$ achieves a ratio strictly less than 9 on the road map $(\sh_n,E_*)$. Then, there must be a strictly positive~$\varepsilon$ such that (we set $x_0 = 0$):
\begin{equation*} 
\forall j \geq 1, \quad \frac{2\sum\limits_{i=1}^jx_i+x_{j-1}+2}{x_{j-1}+2}\leq 9-2\varepsilon. 
\end{equation*}

These inequalities can be rewritten as:
\begin{equation*}
   \forall j \geq 1, \quad \left(\sum_{i=0,i\neq {j-1}}^{j}x_i\right)-(3-\varepsilon)x_{j-1}\leq 8-2\varepsilon.
\end{equation*}
They form a system of linear inequalities with a lower triangular matrix. All entries equal~1 on the diagonal and lower, except for elements on the subdiagonal (lower diagonal) which all equal~$-(3-\varepsilon)$. Our system is $M_{j,\varepsilon} \cdot \ve{x}_j \le \ve{b}_{j,\varepsilon}$ with $\ve{b}_{j,\varepsilon} = (8-2\varepsilon) \cdot \textbf{1}$, where $\mbox{\bf x}_j$ is the vector made of first $j$ values of sequence $(x_i)_{i \ge 1}$ and matrix $M_{j,\varepsilon}$ is the following $j \times j$ matrix:

\begin{equation*}
M_{j,\varepsilon}=\begin{pmatrix}
1 & 0 & 0 & \cdots & 0 \\
-(3-\varepsilon) & 1 & 0 & \cdots & 0 \\
1 & -(3-\varepsilon) & 1 & \cdots & 0 \\
\vdots  & \vdots  & \ddots   & \ddots & \vdots  \\
1 & 1 & 1 & -(3-\varepsilon) & 1 
\end{pmatrix}
\end{equation*}

When no ambiguity is present, we shall omit subscripts and write $M$, $\ve{b}$ and $\ve{x}$. Farkas' lemma (hereafter recalled) deals with the existence (or not) of nonnegative solutions for a system of linear inequalities. Since all values $x_i$'s of vector $\ve{x}$ are at least 1, we may shift our vector by 1 and still be nonnegative. In other words, there should exist a nonnegative vector $\ve{x}' = \ve{x} - \textbf{1}$ such that $M \cdot (\ve{x}' + \textbf{1}) \leq \ve{b}$. After rewriting, we get $M\ve{x}' \leq \ve{b}'$ where $\ve{b}'$ is the vector $\ve{b} - M \cdot \textbf{1}$. Note that $M \cdot \textbf{1}$ has the following coordinates: $(1, -2+\varepsilon, -1+\varepsilon, \varepsilon, 1+\varepsilon, 2+\varepsilon, 3+\varepsilon,\ldots)$. Thus the coordinate $b'_i$ of vector $\ve{b}'$ is negative for any $i \geq 12$. 

\begin{equation*}
\ve{b}' = \ve{b} - M \cdot \textbf{1}=\begin{pmatrix}
7 - 2\varepsilon\\
10 - 3\varepsilon \\
9 - 3\varepsilon \\
8 - 3\varepsilon  \\
\vdots \\
12-j-3\varepsilon
\end{pmatrix}
\end{equation*}

We are now ready to establish a relationship between a system of linear inequalities and the competitiveness of strategies $A_n$.

\begin{proposition}
Assume there exists a positive integer $j$ and some real $\varepsilon > 0$ such that the system $M_{j,\varepsilon} \cdot \ve{x}' \le \ve{b}_{j,\varepsilon}'$, $\ve{x}' \ge \textbf{0}$ has no solution. Then, there exists an integer $n_{j, \varepsilon}$ such that strategy $A_n$ has ratio at least $9-2\varepsilon$.
\label{prop-system}
\end{proposition}
\begin{proof}
From observations above, if all strategies $A_n$ have ratio at most $9-2\varepsilon$ for some $\varepsilon > 0$, then, for any positive integer $j$, the system $M_{j,\varepsilon} \cdot \ve{x}' \le \ve{b}_{j,\varepsilon}'$, $\ve{x}' \ge \textbf{0}$ has necessarily a solution. Using the contraposition for any $\varepsilon > 0$ gives the proof.
\end{proof}

\Cref{prop-system} implies that our lower bound of competitiveness for unit-weighted outerplanar graphs can be proved by showing that some system of linear inequalities has no solution.
We now recall the statement of Farkas' lemma in our context:

\begin{lemma}[Farkas~\cite{farkas1902theorie}, see {\cite[Prop 6.4.3]{Gartner}}]
Exactly one of the following holds:	either the system $M \cdot \ve{x}' \leq \ve{b}' $ has a solution with $\ve{x}' \geq \textbf{0}$, or the system $M^T \cdot \ve{y} \geq \textbf{0}^T$ has a nonnegative solution $\ve{y}$ with $\ve{b}'^T \cdot \ve{y} < 0$.
	\label{lem-farkas}
\end{lemma}
We now find a nonnegative vector $\ve{y}$ of size $j$ such that $M_{j,\varepsilon}^T \cdot \ve{y} \geq \textbf{0}$ and $\ve{b}_{j,\varepsilon}'^T \cdot \ve{y} < 0$, which will allow us to obtain a contradiction. 

\begin{proposition}
For any $\varepsilon > 0$, there exists a positive integer $j$ and a nonnegative vector $\ve{y}$ of size $j$ such that $M_{j,\varepsilon}^T \cdot \ve{y} \geq \textbf{0}$ and $\ve{b}_{j,\varepsilon}'^T \cdot \ve{y} < 0$.
\label{prop-vector}
\end{proposition}
\begin{proof}

For a given $\varepsilon > 0$, we fix an integer $j$ depending on $\varepsilon$. The choice of value $j$ will be made clear hereafter. For now, consider that $j$ is only greater than $12$. Our construction of $\ve{y}$ consists in identifying a solution $M^T \cdot \ve{y}$ very close to vector $\mathbf{0}$ and then verifying whether $\ve{b}'^T \cdot \ve{y} < 0$. 
	
Let us consider the equation $M^T \cdot \ve{y} = \textbf{0}$, where the coordinates of $\ve{y}$ are $y_1,\ldots y_j$. For any $1\le i\le j-1$, we have:
$$ 
y_i-(3-\varepsilon)y_{i+1}+y_{i+2}+y_{i+3}+...+y_{j}=0.
$$
By subtracting two consecutive such equations, we get that, for any $1\le i\le j-3$:
$$
y_i = (4-\varepsilon) y_{i+1} - (4-\varepsilon) y_{i+2}. 
$$
Reversing the indices, we recognize a linear recurrence relation of depth 2: $u_{n+2} = (4-\varepsilon) u_{n+1} - (4-\varepsilon) u_n$. The characteristic equation is $\lambda^2-(4-\varepsilon)\lambda+(4-\varepsilon)=0$ with roots:

	$$
	\lambda_{1,2}=\frac{(4-\varepsilon)\pm\sqrt{-\varepsilon(4-\varepsilon)}}{2}.
	$$
	
	Observe that if $\varepsilon = 0$, then the characteristic equation has a single root and the sequence is exponentially increasing.
	For $\varepsilon > 0$, however, both roots are complex numbers.

	Since there are two roots, we could rewrite it as $
	u_n=(\sqrt{4-\varepsilon})^n(c_1\cos(\alpha n)+ c_2\sin(\alpha n))$, 
	where $\alpha=\atan\left(\sqrt{\frac{\varepsilon}{4-\varepsilon}}\right)$. Note that when $\varepsilon$ tends to zero, the period of oscillation tends to infinity. The coordinates of the vector $\ve{y}$ will follow, starting from the bottom, the scheme of this sequence $u_n$. We fix as an initial condition $y_j = 1$.
	
	Consequently, $y_{j-1} = 3-\varepsilon$ and the following terms 
	follow the scheme of the 
	sequence. With these initial conditions, we obtain:
	
	\begin{equation}
	y_{j-p}=\frac{2(\sqrt{4-\varepsilon})^{p-1}}{\sqrt{\varepsilon}}\cos(\alpha p-\beta).
	~~~~\ve{y} =
	\begin{pmatrix}
    y_1\\
    \vdots \\
    y_{j-2} \\
    y_{j-1} \\
    y_{j}
    \end{pmatrix} =
	\begin{pmatrix}
    u_{j-1}\\
    \vdots \\
    u_2 \\
    u_1 \\
    u_0
    \end{pmatrix} =
    \begin{pmatrix}
    u_{j-1}\\
    \vdots \\
    (3-\varepsilon)^2 - 1 \\
    3-\varepsilon \\
    1
    \end{pmatrix}
	\label{eq-sequence}
	\end{equation}
	where $\beta=\atan\left(\frac{2-\varepsilon}{\sqrt{\varepsilon(4-\varepsilon)}}\right).$ In this way, $M^T \cdot \ve{y}$ is the vector made of zeros except the last coordinate being equal to 1.
	
	Equation~\eqref{eq-sequence} shows us that we start from positive values ($y_j, y_{j-1}, \ldots$) and then alternate between positive and negative with an increasing amplitude. For any small $\varepsilon > 0$, we select $j$ such that terms $u_0,\ldots,u_{j-12}$ are positive and the last eleven terms $u_{j-11},\ldots,u_{j-1}$ are negative. This choice of $j$ is possible since, for sufficiently small $\varepsilon$, the ``angular speed'' $\alpha$ and the ``shift'' $\beta$ are negligible compared to $\frac{\pi}{2}$. Hence, let $j$ be the integer such that $\alpha(j-11) - \beta \le \frac{\pi}{2}$, while $\alpha(j-10) - \beta > \frac{\pi}{2}$. On vector $\ve{y}$, it means that the $11$ first terms $y_1,\ldots,y_{11}$ are negative while all others are positive. Observe that the more $\varepsilon$ decreases towards $0$, the larger the period of sequence $(u_i)$ is and hence the larger this integer $j$ is. 
	
	Values $y_1,\ldots,y_{11}$, which are the negative values of $\ve{y}$, are replaced by 0: we obtain a new vector $\ve{y}'$. In this way, value $\ve{b}'^T \cdot \ve{y}'$ is negative: first values $y_1',\ldots,y_{11}'$ are zeros, second values $y_{12}',\ldots,y_j'$ are positive while $b_{12}',\ldots,b_j'$ are negative. Moreover, obviously, $\ve{y}' \ge \ve{0}$. Let us check that even with this modification on $11$ entries, we still have $M^T \cdot \ve{y}' \geq \textbf{0}$, allowing us to use Farkas' lemma.
	
	Suppose by way of contradiction that some coordinate of $M^T \cdot \ve{y}'$ is negative: say $y_i' - (3-\varepsilon) y'_{i+1} + y'_{i+2} + \ldots + y'_j < 0.$ As all $y_i'$ are nonnegative, then necessarily $y_{i+1}'$ is positive and thus $y_{i+1}' = y_{i+1}$. Either $y_i' > 0$, so $y_i' = y_i$ and the linear sum is nonnegative from the definition of sequence $(u_i)_{i\ge 1}$, a contradiction; or $y_i' = 0$ but then again the previous equation is positive since $y_i - (3-\varepsilon) y_{i+1}' + y_{i+2}' + \ldots + y_j' = 0$ and we shifted it positively by replacing $y_i$ by $y_i'$ which is greater. In brief, vector $\ve{y}'$ verifies $M^T \cdot \ve{y}' \geq 0$. 
	
	As a conclusion, we identified a nonnegative vector $\ve{y}'$ satisfying the requirements of the proposition.
	\end{proof}
	
	By Farkas's lemma, we conclude that the initial system of Proposition~\ref{prop-system} of inequalities does not have any solution. So there is no strategy with competitive ratio less than 9 on shell graphs.
	
	\medskip

	\lowerbound*

	\begin{proof}. Simply fix $\varepsilon ' = 2\varepsilon > 0$. From \Cref{lem-farkas} and \Cref{prop-vector}, we know that there exists an integer $j$ such that there is no nonnegative solution $\ve{x}'$ of system $M_{j,\varepsilon} \cdot \ve{x}' \le \ve{b}_{j,\varepsilon}'$. From \Cref{prop-system}, there is no deterministic strategy achieving ratio $9-\varepsilon'$ on every road map of the family of graphs $\sh_n$, and hence on the super-family of unit-weighted outerplanar graphs. 
	\end{proof}

\section{The case of arbitrarily weighted outerplanar graphs} \label{sec-weighted}

Given our results on the unit-weighted case (which give as an easy corollary ratio $9S$ for fixed stretch $S$), a natural question is whether we can design a deterministic strategy achieving a constant competitive ratio for the more general family of arbitrarily weighted outerplanar graphs. In this section, we prove that this is impossible and that there are weighted outerplanar graphs on which the competitive ratio obtained is necessarily greater than some function $g(k) \in \Omega(\frac{\ln k}{\ln \ln k})$. 

We introduce a sub-family of outerplanar graphs that will be useful in this section.

\begin{definition}
    Given an outerplanar graph $G$, two of its vertices $s$ and $t$, and an embedding of $G$ on the plane, $G$ is said to be \emph{$(s,t)$-unbalanced} if either it is the single edge $st$, or it is 2-connected and one of its sides contains only $s$ and $t$.
    \label{def-unbalanced}
\end{definition}

In other words, an $(s,t)$-unbalanced outerplanar graph is such that $s$ and $t$ are neighbors on the outer face. Assume without loss of generality that the lower side contains all vertices and that the upper side only contains $s$ and $t$ and simply consists of a single edge $st$. Note that such a graph does not have any vertical chord. We show in the remainder that some competitive ratio $g(k)$ cannot be obtained even on weighted $(s,t)$-unbalanced outerplanar graphs. 

From now on, we fix some positive real value $\varepsilon^* > 0$. This value will be assigned as an edge weight and can be made as small as needed. All other edge weights will be positive integers.

We begin with the definition of a graph transformation $\mathcal{T}$ which takes as input a weighted $(s,t)$-unbalanced outerplanar graph $H=(V,E,\omega)$, and two integers $S$  and $N$. 
The construction of the output graph $\mathcal{T}(H,S,N, \varepsilon^*)$ works as follows:
\begin{itemize}
    \item Create two vertices $s$ and $t$ with an edge $st$ of weight $S$. This edge will stand as the upper side of the graph.
    \item Create $N$ copies of the graph $H$. These copies are denoted by $H^{(1)},\ldots,H^{(N)}$ and the source/target pair of each $H^{(j)}$ is denoted by $(s_j,t_j)$.
    \item Connect in series all copies $H^{(1)},\ldots,H^{(N)}$ from $s$ to $t$ in order to form the lower side of the graph, using their source/target as input/output vertices. That is, merge $s$ with $s_1$, $t_i$ with $s_{i+1}$ for $i \in \{1,\ldots,N-1\}$. Add an edge (called \emph{terminus}) of weight $\varepsilon^*$ between vertices $t_N$ and $t$.
    \item Add all edges $t_j t$ for $1\le j\le N-1$ with weight $\varepsilon^*$.
\end{itemize}

\Cref{fig-yanGraph} illustrates the graph $\mathcal{T}(H,S,N, \varepsilon^*)$ obtained. Observe that, since $H$ is $(s,t)$-unbalanced outerplanar, $\mathcal{T}(H,S,N, \varepsilon^*)$ is outerplanar. Furthermore, it clearly is $(s,t)$-unbalanced outerplanar: the lower side of each copy of $H$ contains all its own vertices, so all vertices are on the lower side, and the edge $st$ exists. We fix $t_0 = s$.

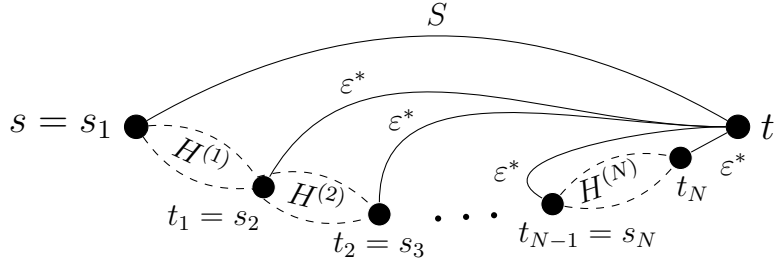
\begin{figure}[h]
	\centering
	\scalebox{1}{\input{tikzjournal/grapheDeYanstar}}
	\caption{The graph $\mathcal{T}(H,S,N)$ with its outerplanar embedding.}
	\label{fig-yanGraph}
\end{figure}

We next define several integer sequences that will be useful for our competitive analysis. 

\begin{definition}[Integer sequences] The following sequences are defined for integer indices $i\ge 1$, given some $\varepsilon>0$.
\begin{itemize}
    \item the \emph{ratio sequence}: $r^\varepsilon_i = i + 1 - \varepsilon$,
    \item the \emph{rounded ratio sequence} $r_i= i+1$,
    \item the \emph{weight sequence}, recursively: $S_1 = 2$ and $S_i = S_{i-1}( r_{i-1}  + 1)$, so $S_i = (i+1)!$,
    \item the \emph{copy sequence}: $N_i = \frac{S_i r_{i-1}}{S_{i-1}} = i(i+1)$,
    \item the \emph{blockage sequence}, recursively: $k_1 = 1$ and $k_i = N_i(k_{i-1}+1)$.
\end{itemize}
\label{def:sequences}
\end{definition}

Note that the blockage sequence is asymptotically close to $(i!)^2$, as $N_i = i(i+1)$, but with an extra term. Inductively, we can show a rough upper bound for $k_i$: $k_i \le ((i+1)!)^2$. It works for $i=1$ and, by induction, $k_i \le  i(i+1)(i!)^2 + i(i+1) \le (i!)((i+1)!)\left(i + \frac{i(i+1)}{i!(i+1)!}\right) \le ((i+1)!)^2$.

We present our main technical statement for this section, which will be used to state a lower bound of competitiveness at the end of the section.

\begin{proposition}
\label{prop-yanFamilyOfCounterExamples}
For any nonnegative integer $i$ and any $\varepsilon>0$, there exists a family $\mathcal{R}_i$ of road maps which satisfies the following properties:
\begin{enumerate}
    \item all the road maps of $\mathcal{R}_i$ are defined on the same weighted  $(s,t)$-unbalanced outerplanar graph denoted by $H_i$,
    \item each road map of $\mathcal{R}_i$ has at most $k_i < ((i+1)!)^2$ blockages,
    \item given any deterministic strategy $A$ on $H_i$, there exists a blockage configuration $E_*$ such that: 
    \begin{itemize}
        \item $(G,E_*) \in \mathcal{R}_i$,
        \item the distance traversed $\dtr{A}(G,E_*)$ is at least $S_i = (i+1)!$,
        \item the competitive ratio of $A$ on $(G,E_*)$ is at least $r^\varepsilon_i = i + 1 - \varepsilon$.
    \end{itemize}
\end{enumerate}
\end{proposition}
\begin{proof}
We proceed by induction on $i$. In the induction step, we will focus on some graph $H_{i} = \mathcal{T}(H_{i-1},S_i,N_i, \varepsilon^*)$ and establish a trade-off on the value of the optimal distance between $s$ and the last visited $t_j$ vertex on the lower side.

\medskip
\textit{Base case}. Let $H_1$ be the weighted graph with three vertices $s$, $u$ and $t$, an edge $st$ of weight $2$, an edge $su$ of weight $1$ and an edge $ut$ of weight $\varepsilon^*$. The graph $H_1$ is clearly $(s,t)$-unbalanced as it is outerplanar and the lower side contains the three vertices. We fix $k_1 = 1$. There are two road maps in $\mathcal{R}_1$: either nothing is blocked, or only the edge $ut$ is blocked. On this very restricted graph, there are two strategies for the traveller. The first one is to simply traverse the edge $st$, but if nothing is blocked on the lower side, the competitive ratio obtained is $\frac{2}{1+\varepsilon^*}=2 - \frac{2\varepsilon^*}{1+\varepsilon^*}$, which can be made greater than $r^\varepsilon_1 = 2-\varepsilon$ by choosing $\varepsilon^*$ small enough compared to $\varepsilon$. The second strategy is to traverse the edge $su$ instead of $st$, but $ut$ can be blocked. In this case, the traveller goes back to $s$ and then traverses the open edge $st$ which is the optimal offline path. The competitive ratio is thus exactly $\frac{4}{2} = 2 \geq r^\varepsilon_1$. In summary, it is not possible to achieve a competitive ratio strictly smaller than $r^{\varepsilon}_1$ on $H_1$. Moreover, given any of the two possible strategies, there is a blockage configuration which forces us to walk with a distance of at least $2 = S_1$.

\textit{Induction step}. We assume that the statement above holds for some integer $i-1 \ge 1$. Let $H_{i-1}$ be the $(s,t)$-unbalanced outerplanar graph on which the road maps of $\mathcal{R}_{i-1}$ are defined. We will construct the graph $H_{i}$ by applying the transformation $\mathcal{T}$ on graph $H_{i-1}$. Based on the sequences of \Cref{def:sequences}, we define $H_i = \mathcal{T}(H_{i-1},S_i,N_i, \varepsilon^*)$.

\begin{figure}[h]
	\centering
	\scalebox{1}{\input{tikzjournal/exampleH2star}}
	\caption{Example of the graph $H_2$ with $N_2 = 6$ copies of $H_1$ on the lower side.}
	\label{fig-exampleH2}
\end{figure}
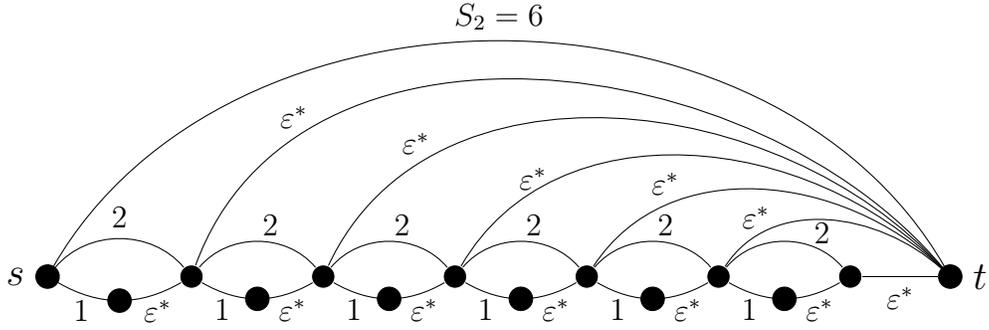

Together with $H_i$, we consider sets of blocked edges $E_*$ which contain some edges $t_j t$ but also some edges lying inside the copies $H_{i-1}^{(j)}$. However, $st$ will never be in $E_*$. The road maps of $\mathcal{R}_i$ are all the road maps obtained by considering any combination of blockages for each copy of $H_{i-1}$ inside $\mathcal{R}_{i-1}$. Said differently, each copy $H_{i-1}^{(j)}$ together with its blockages form potentially any road map of $\mathcal{R}_{i-1}$. Moreover, edges $t_jt$ can also be blocked. All the settings consistent with this description give us the collection $\mathcal{R}_i$.

Note that, as a consequence of this construction, each time the traveller traverses entirely a copy $H_{i-1}^{(j)}$ (from $s_j$ to $t_j$), there exists a blockage set $E_*$ making him traverse a distance at least $r^\varepsilon_{i-1}$ times the shortest open $(s_j,t_j)$-path, from the induction hypothesis. Furthermore, the number of blockages is at most $N_i(k_{i-1}+1)$, as there are $N_i$ copies and also $N_i$ edges $t_jt$ with $1\le j \le N$ (including the terminus). Hence, each road map of $\mathcal{R}_i$ contains at most $k_i$ blocked edges.
In order to prove the third statement of our proposition, we distinguish three types of strategies.

\smallskip\noindent\textbf{Go to terminus.} A first (but naive) strategy for the traveller could be to try to reach $t$ without traversing the edge $st$, {\em i.e.}, by going only through the lower side. Obviously, the competitiveness is strongly decreased if not only the terminus $t_N t$ is blocked but also all chords $t_j t$ with $j\ge 1$. In this case, the traveller has to go back to $s$ and traverse $st$. Considering this worst blockage set for our strategy, the distance traversed is at least $\dtr{} + d + S_i$, where $\dtr{}$ is the distance traversed in the outward journey from $s$ to $t_N$ and $d$ is the cost of the shortest open path from $s$ to $t_N$. From the induction hypothesis, we have $\dtr{} \ge S_{i-1}N_i$, as we must traverse $N_i$ copies of $H_{i-1}$ in series. The optimal offline cost of the road map is $S_i$, since there is a single open $(s,t)$-path (the edge $st$ itself). In this scenario, the competitive ratio $c$ of such strategy is at least:
\[
c \ge \frac{\dtr{} + d + S_i}{S_i} > 1 + \frac{\dtr{}}{S_i} \ge 1 + N_i\frac{S_{i-1}}{S_i} = 1 + r_{i-1} = i+1.
\]
The last simplification comes from the definition of the sequence $N_i$. Hence, such a strategy leads to a competitive ratio at least $r_i = i + 1$. We next focus on different strategies where the lower side is only partially explored.

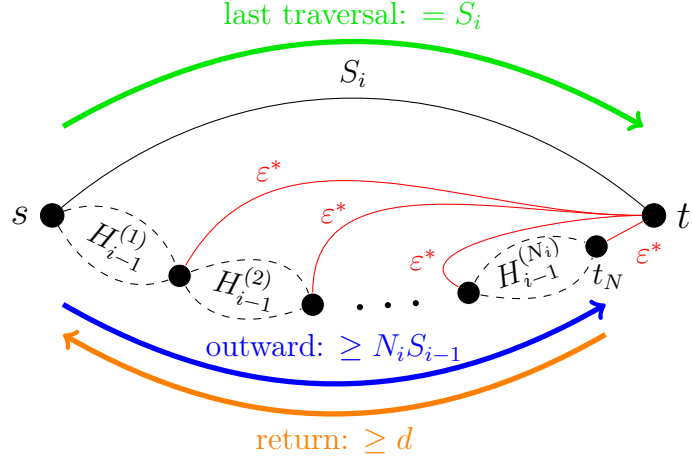
\begin{figure}[t]
	\centering
	\scalebox{1}{\input{tikzjournal/trajetYanstar}}
	\caption{Lower bound on the distance traversed over graph $H_i$ when we go to the terminus edge: the outward journey to $t_N$ (in blue), the return journey to $s$ (in orange) and the traversal of edge $st$ (in green). All edges $t_j t$ with $1\le j\le N$ are blocked (in red).}
	\label{fig-trajetYan}
\end{figure}

\smallskip\noindent\textbf{Explore deeply the lower side.} A second strategy consists in exploring the lower side with a significant distance (but without reaching the terminus), coming back to $s$ if no edge $t_j t$ allowed us to reach $t$, and then traverse $st$. Let $t_q$ be the last visited $t_j$ on the lower side (said differently, the highest value of $j$ such that $t_j$ has been explored by the traveller): we have $0\le q < N$. Let $d$ be the cost of the shortest open path from $s$ to $t_q$. We suppose that $d >S_i - S_{i-1}$. The reverse inequality will be treated later.

The distance traversed by the traveller is at least $r^{\varepsilon}_{i-1}d + d + S_i$. Indeed, the cost of the outward journey towards $t_q$ is at least the optimal open distance multiplied by $r^{\varepsilon}_{i-1}$ from induction hypothesis, as we traverse copies of $H_{i-1}$ in series. Concerning the optimal offline cost of our road map, we have $\dopt \le S_i$ since edge $st$ is open in any instance of $\mathcal{R}_i$. The competitive ratio $c$ of this strategy is thus at least:

\begin{align*}
	c \ge \frac{r^{\varepsilon}_{i-1}d + d + S_i}{S_i} &  > 1 + (r^{\varepsilon}_{i-1} + 1)\frac{S_i-S_{i-1}}{S_i} \\
	& = r^{\varepsilon}_{i-1} + 1 + 1 - (r^{\varepsilon}_{i-1} + 1)\frac{S_{i-1}}{S_i} \\
	& = r^{\varepsilon}_{i-1} + 1 + 1 - \frac{(r^{\varepsilon}_{i-1} + 1)}{(r_{i-1} + 1)} \\
	& \ge r^{\varepsilon}_{i-1} + 1.
\end{align*}

\smallskip\noindent\textbf{Explore a bit the lower side.} The last strategy consists in exploring the lower side with a small budget (or not exploring it at all!). This last case corresponds to the inequality $0\le d\le S_i - S_{i-1}$, completing the analysis of the previous strategy. We keep $t_q$ as the notation for the last visited $t_j$. A possible setting of the blockages is the following: the edges $t_q t_{q+1}$ and $t_{q+1} t$ are open. In this scenario, $\dopt \le d + S_{i-1} +\varepsilon^*$. Indeed, there is an open $(s,t)$-path with cost $d + S_{i-1} +\varepsilon^*$: go from $s$ to $t_q$ with the shortest open $(s,t_q)$-path, traverse chord $t_q t_{q+1}$ of weight $S_{i-1}$ and finally chord $t_{q+1} t$. The competitive ratio $c$ of our strategy satisfies, provided that  $\varepsilon^*\leq \varepsilon/2$ :

	\begin{align*}
		c \ge \frac{r^{\varepsilon}_{i-1}d + d + S_i}{d + S_{i-1}+\varepsilon^*} &  = \frac{(r^{\varepsilon}_{i-1} + 1)(d + S_{i-1} +\varepsilon^*) - (r^{\varepsilon}_{i-1} + 1)(S_{i-1}+\varepsilon^*) + S_i}{d + S_{i-1}+\varepsilon^*} \\
		& = r^{\varepsilon}_{i-1} + 1 + \frac{S_i - (r^{\varepsilon}_{i-1} + 1)(S_{i-1}+\varepsilon^*)}{d+S_{i-1}+\varepsilon^*} \\
		& = r^{\varepsilon}_{i-1} + 1 + \frac{S_{i-1}( r_{i-1}  + 1) - (r^{\varepsilon}_{i-1} + 1)(S_{i-1}+\varepsilon)}{d+S_{i-1}+\varepsilon^*}\\
		& = r^{\varepsilon}_{i-1} + 1 + \frac{S_{i-1}(r_{i-1}-r^{\varepsilon}_{i-1}) - \varepsilon^*(r^{\varepsilon}_{i-1}+1)}{d+S_{i-1}+\varepsilon^*} \\
		& = r^{\varepsilon}_{i-1} + 1 + \frac{S_{i-1}\varepsilon- \varepsilon^*(r^{\varepsilon}_{i-1}+1)}{d+S_{i-1}+\varepsilon^*} \\
		& \geq r^{\varepsilon}_{i-1} + 1 + \frac{\varepsilon}{2}\frac{2S_{i-1}- (r^{\varepsilon}_{i-1} +1)}{d+S_{i-1}+\varepsilon^*} \\
		& \ge r^{\varepsilon}_{i-1} + 1.
	\end{align*}

The last simplification comes from the definition of the sequence $S_i$ and the fact that $2S_{i-1}=2(i+1)!\geq i+1 + \varepsilon$.

\smallskip
These three types of strategies cover all the different possibilities for the traveller to traverse $H_i$. In each case, there is a blockage set of $\mathcal{R}_i$ which makes the competitive ratio be at least $r^{\varepsilon}_i = r^{\varepsilon}_{i-1} + 1$. Observe that, in each of these worst cases, the distance traversed is at least $r^{\varepsilon}_{i-1}d + d + S_i \ge S_i$. This completes the induction step, and thus the third statement holds.
\end{proof}

The consequence of \Cref{prop-yanFamilyOfCounterExamples} is that for any parameter $k\ge 1$, there exists some family of road maps $\mathcal{R}$ for the \kctp on which competitive ratio $r$ cannot be attained, with $r$ being the largest integer such that $k \ge (r!)^2$. The proof of this claim comes by simply considering the set of instances $\mathcal{R}_i$, where $i$ is the largest integer such that $k \ge ((i+1)!)^2$. Hereafter, we express the ratio $r$ in function of $k$.

\weightedBound*

\begin{proof}
    Let $i$ be the largest positive integer such that $k \ge ((i+1)!)^2 \ge k_i$. On the collection of road maps $\mathcal{R}_i$, for every $\varepsilon > 0$, no  strategy can obtain competitive ratio $r^\varepsilon_i = i+1-\varepsilon$ and all these road maps contain less than $k_i \le ((i+1)!)^2\le k$ blockages. From the definition of $i$, we have $k < ((i+2)!)^2$ and hence $\ln k < 2\ln ((i+2)!) < 2\ln ((i+2)^{i+2}) = 2\ln (i+2) e^{\ln (i+2)}$. We obtain:
    \[
    \frac{\ln k}{2} < \ln (i+2) e^{\ln (i+2)}
    \]
    The Lambert W function is the converse function of $x \rightarrow xe^x$ (more details in~\cite{CoGo96}). We associate $\ln (i+2)$ to the entry of this function, and we have:
    \[
    \ln (i+2) > W\left(\frac{\ln k}{2}\right)
    \]
    Going back to the definition of the sequence $r_i$, we have $r_i > e^{W(\frac{\ln k}{2})} - 1$. So, for any integer $k \ge 1$, it is possible to identify a collection of weighted outerplanar graphs on which ratio $e^{W(\frac{\ln k}{2})} - 1$ cannot be attained.
\end{proof}

\section{Perspectives} \label{sec-conclusion}

We highlighted a non-trivial unit-weighted family of graphs (outerplanar) for which there exists a deterministic strategy with constant competitive ratio 9, which is optimal. However, we proved that no constant competitive ratio can be achieved for arbitrarily weighted outerplanar graphs. Several questions arise.

Since some sub-families of outerplanar graphs have constant competitive ratio in the weighted case (trees and cycles, which imply cacti from \Cref{lem-articulationPoints}) while a very close super-family admits the general bound $2k+1$ in the unit-weighted case (planar of treewidth~$2$), a natural question is to investigate where the competitive gaps lie in both cases.
For the unit-weighted case, future research could focus on the natural extension of $p$-outerplanar graphs~\cite{Ba94}, with $p$ successive outer faces, in order to generalize constant competitiveness.
Coming back to arbitrarily weighted outerplanar graphs, a natural question is whether one can achieve a competitive ratio $O(\log k)$ on this family of graphs. Such a positive result would close almost entirely the competitiveness gap for outerplanar graphs, thanks to our lower bound (Theorem~\ref{thm:weightedbound}).

To achieve constant competitive ratio on arbitrarily weighted graphs, a good candidate could be graphs with bounded-sized minimal edge $(s,t)$-cuts, for which ratio $\sqrt{2}k + O(1)$ is known~\cite{BeSa23}. Observe that our construction $\mathcal{T}$ 
increases the size of edge $(s,t)$-cuts.
We conjecture that there exists a polynomial-time deterministic strategy achieving constant competitive ratio on graphs with edge $(s,t)$-cuts of bounded size.

\bibliographystyle{plain}
\bibliography{ctp}

\end{document}

%% file: tikzjournal/outerplanar-exemple.tex
\begin{tikzpicture}

    \newcommand\Square[1]{+(-#1,-#1) rectangle +(#1,#1)}

    \newcommand\clipbetween[2]
        {%
            \coordinate (c) at (barycentric cs:#1=1,#2=1);
            \clip (c) let \p1=(#1), \p2=(#2) in node{\n1} \Square{{veclen(\x2-\x1,0)/2}};
        }

    \node[noeud,scale=.7] (s) at (0,0) {};
    \draw (s) node[left,scale=1.5] {$s$};
    \node[noeud,scale=.7] (t) at (6,0) {};
    \draw (t) node[right,scale=1.5] {$t$};

    \begin{scope}
        \draw[bend left,opacity=0]
        (s)to 
        node[noeud,pos=0.1,opacity=1,scale=.5](p1){} 
        node[noeud,pos=0.2,opacity=1,scale=.5](p2){} 
        node[noeud,pos=0.3,opacity=1,scale=.5](p3){}
        node[noeud, pos=0.8, opacity=1, scale=.5](pk1){}
        node[noeud, pos=0.9, opacity=1, scale=.5](pk){}
        (t);

   \end{scope}

   \begin{scope}
        \clipbetween{s}{p3}
        \draw[bend left,opacity=1]
        (s)to 
        (t);
   \end{scope}

   \begin{scope}
        \clipbetween{p2}{pk1}
        \draw[bend left,opacity=1, dash pattern=on 3pt off 3pt]
        (s)to 
        (t);
   \end{scope}

   \begin{scope}
        \clipbetween{pk1}{t}
        \draw[bend left,opacity=1]
        (s)to 
        (t);
   \end{scope}
   
    \begin{scope}
        \draw[bend right,opacity=0]
        (s)to 
        node[noeud,pos=0.1,opacity=1,scale=.5](ql){} 
        node[noeud,pos=0.2,opacity=1,scale=.5](ql1){}
        node[noeud,pos=0.35,opacity=1,scale=.5](ql2){}
        node[noeud,pos=0.65,opacity=1,scale=.5](q4){} 
        node[noeud,pos=0.75,opacity=1,scale=.5](q3){} 
        node[noeud,pos=0.85,opacity=1,scale=.5](q2){} 
        node[noeud,pos=0.95,opacity=1,scale=.5](q1){}
        (t);
   \end{scope}

    \begin{scope}
        \clipbetween{s}{ql2}
        \draw[bend right,opacity=1]
        (s)to 
        (t);
   \end{scope}

    \begin{scope}
        \clipbetween{ql2}{q3}
        \draw[bend right,opacity=1, dash pattern=on 3pt off 3pt]
        (s)to 
        (t);
   \end{scope}

    \begin{scope}
        \clipbetween{q3}{t}
        \draw[bend right,opacity=1]
        (s)to 
        (t);
   \end{scope}

    \draw (p1) node[above=5 pt, scale=1.2] {$p_1$};
    \draw (p2) node[above=5 pt, scale=1.2] {$p_2$};
    \draw (p3) node[above=5 pt, scale=1.2] {$p_3$};
    \draw (pk1) node[above=5 pt, scale=1.2] {$p_{h-1}$};
    \draw (pk) node[above right=2 pt, scale=1.2] {$p_h$};

    \draw (q1) node[below=5 pt, scale=1.2] {$q_1$};
    \draw (q2) node[below=5 pt, scale=1.2] {$q_2$};
    \draw (q3) node[below=5 pt, scale=1.2] {$q_3$};
    \draw (q4) node[below=5 pt, scale=1.2] {$q_4$};
    \draw (ql) node[below=5 pt, scale=1.2] {$q_\ell$};
    \draw (ql1) node[below=5 pt, scale=1.2] {$q_{\ell-1}$};
    \draw (ql2) node[below=5 pt, scale=1.2] {$q_{\ell-2}$};
    
    \draw (p2) -- (ql);
    \draw (p2) -- (ql1);
    \draw (p3) -- (ql2);
    \draw (q4) -- (pk1);

    \draw[bend left] (q3) to (q1);
    \draw[bend left] (q4) to (q1);

\end{tikzpicture}

%% file: tikzjournal/westphal.tex
\begin{tikzpicture}

\node[noeud] (s) at (1,4) {};
\node[scale=1.8] at (0.6,4) {$s$};
\node[noeud] (v1) at (6,1) {};
\node[noeud] (v2) at (6,2.5) {};
\node[noeud] (v3) at (6,4) {};
\node[noeud] (v4) at (6,5.5) {};
\node[noeud] (v5) at (6,7) {};
\node[noeud] (t) at (11,4) {};
\node[scale=1.8] at (11.4,4) {$t$};

\foreach \I in {1,...,5} {
    \draw[free] (s) to node[pos=0.75,above,scale=1.4] {1} (v\I);
}
\foreach \I in {1,2,3,5} {
    \draw[blocked] (t) to node[pos=0.75,above,scale=1.4,black] {$\varepsilon$} (v\I);
}
\draw[free] (v4) to node[pos=0.25,above,scale=1.4] {$\varepsilon$} (t);

\end{tikzpicture}

%% file: tikzjournal/unweighted-decomposition.tex
\begin{tikzpicture}
	
	\node (G) at (0,7) {
		\begin{tikzpicture}
			\node[noeud] (s) at (0,1) {};
			\node[noeud] (u1t) at (1,2) {};
			\node[noeud] (u2t) at (2,2) {};
			\node[noeud] (u1b) at (1,0) {};
			\node[noeud] (u2b) at (2,0) {};
			\node[noeud] (t1) at (3,1) {};
			
			\draw (s)to(u1t);
			\draw (s)to(u1b);
			\draw (u1t)to(u2t);
			\draw (u1b)to(u2b);
			\draw (u2t)to(t1);
			\draw (u2b)to(t1);
			\draw (u1t)to(t1);
			\draw (u1b)to(t1);
			
			\node[noeud] (t2) at (4,1) {};
			
			\draw (t1)to(t2);
			
			\node[noeud] (v1t) at (5,2) {};
			\node[noeud] (v2t) at (6,2) {};
			\node[noeud] (v3t) at (7,2) {};
			\node[noeud] (v4t) at (8,2) {};
			\node[noeud] (v5t) at (9,2) {};
			\node[noeud] (v1b) at (5,0) {};
			\node[noeud] (v2b) at (6,0) {};
			\node[noeud] (v3b) at (7,0) {};
			\node[noeud] (v4b) at (8,0) {};
			\node[noeud] (v5b) at (9,0) {};
			\node[noeud] (t) at (10,1) {};
			
			\draw (t2)to(v1t);
			\draw (v1t)to(v2t);
			\draw (v2t)to(v3t);
			\draw (v3t)to(v4t);
			\draw (v4t)to(v5t);
			\draw (t2)to(v1b);
			\draw (v1b)to(v2b);
			\draw (v2b)to(v3b);
			\draw (v3b)to(v4b);
			\draw (v4b)to(v5b);
			\draw (v5t)to(t);
			\draw (v5b)to(t);
			\draw (v3t)to(v3b);
			\draw (v4t)to(v4b);
			\draw (v5t)to(v4b);
			\draw (v1t)to(v3b);
			
			\draw (s) node[left,scale=2] {$s$};
			\draw (t) node[right,scale=2] {$t$};
			
			\node[noeud] (i01) at (1,3) {};
			\draw (u1t)to(i01);
			\node[noeud] (i02) at (2,-1) {};
			\node[noeud] (i03) at (1,-2) {};
			\node[noeud] (i04) at (3,-2) {};
			\node[noeud] (i05) at (2,-3) {};
			\draw (u2b)to(i02)to(i03)to(i05)to(i04)to(i02);
			\node[noeud] (i06) at (5,-1) {};
			\node[noeud] (i07) at (4,-2) {};
			\node[noeud] (i08) at (5,-2) {};
			\node[noeud] (i09) at (6,-2) {};
			\node[noeud] (i10) at (4.5,-3) {};
			\draw (v1b)to(i06)to(i07)to(i10)to(i08)to(i09)to(i06)to(i08);
			\node[noeud] (i11) at (5,3) {};
			\node[noeud] (i12) at (6,4) {};
			\node[noeud] (i13) at (7,3) {};
			\draw (v2t)to(i11)to(i12)to(i13)to(v2t);
			\node[noeud] (i14) at (7,-1) {};
			\node[noeud] (i15) at (8,-2) {};
			\node[noeud] (i16) at (9,-1) {};
			\draw (v4b)to(i14)to(i15)to(i16)to(v4b);
			\node[noeud] (i17) at (11,0) {};
			\draw (t)to(i17);
			\node[noeud] (i18) at (11,2) {};
			\node[noeud] (i19) at (12,3) {};
			\node[noeud] (i20) at (13,3) {};
			\node[noeud] (i21) at (14,2) {};
			\node[noeud] (i22) at (13,1) {};
			\node[noeud] (i23) at (12,1) {};
			\draw (t)to(i18)to(i19)to(i20)to(i21)to(i22)to(i23)to(i18);
			\node[noeud] (i24) at (12,-1) {};
			\node[noeud] (i25) at (13,0) {};
			\draw (i17)to(i23)to(i25)to(i24)to(i17);
		\end{tikzpicture}
	};
	
	\node (Gdec) at (0,0) {
		\begin{tikzpicture}
			\node[noeud] (s) at (0,1) {};
			\node[noeud] (u1t) at (1,2) {};
			\node[noeud] (u2t) at (2,2) {};
			\node[noeud] (u1b) at (1,0) {};
			\node[noeud] (u2b) at (2,0) {};
			\node[noeud] (t1) at (3,1) {};
			
			\draw (s)to(u1t);
			\draw (s)to(u1b);
			\draw (u1t)to(u2t);
			\draw (u1b)to(u2b);
			\draw (u2t)to(t1);
			\draw (u2b)to(t1);
			\draw (u1t)to(t1);
			\draw (u1b)to(t1);
			
			\node[noeud] (s2) at (4,1) {};
			\node[noeud] (t2) at (5,1) {};
			
			\draw (s2)to(t2);
			
			\node[noeud] (s3) at (6,1) {};
			\node[noeud] (v1t) at (7,2) {};
			\node[noeud] (v2t) at (8,2) {};
			\node[noeud] (v3t) at (9,2) {};
			\node[noeud] (v4t) at (10,2) {};
			\node[noeud] (v5t) at (11,2) {};
			\node[noeud] (v1b) at (7,0) {};
			\node[noeud] (v2b) at (8,0) {};
			\node[noeud] (v3b) at (9,0) {};
			\node[noeud] (v4b) at (10,0) {};
			\node[noeud] (v5b) at (11,0) {};
			\node[noeud] (t) at (12,1) {};
			
			\draw (s3)to(v1t);
			\draw (v1t)to(v2t);
			\draw (v2t)to(v3t);
			\draw (v3t)to(v4t);
			\draw (v4t)to(v5t);
			\draw (s3)to(v1b);
			\draw (v1b)to(v2b);
			\draw (v2b)to(v3b);
			\draw (v3b)to(v4b);
			\draw (v4b)to(v5b);
			\draw (v5t)to(t);
			\draw (v5b)to(t);
			\draw (v3t)to(v3b);
			\draw (v4t)to(v4b);
			\draw (v5t)to(v4b);
			\draw (v1t)to(v3b);
			
			\draw (s) node[left,scale=2] {$s=s_1$};
			\draw (t1) node[above,scale=1.75,yshift=0.5mm] {$t_1$};
			\draw (s2) node[above,scale=1.75,yshift=0.5mm] {$s_2$};
			\draw (t2) node[above,scale=1.75,yshift=0.5mm] {$t_2$};
			\draw (s3) node[above,scale=1.75,yshift=0.5mm] {$s_3$};
			\draw (t) node[right,scale=2] {$t_3=t$};
			
			\draw[rounded corners,dashed] (2.75,0.75) rectangle (4.25,1.25);
			\draw[rounded corners,dashed] (4.75,0.75) rectangle (6.25,1.25);
		\end{tikzpicture}
	};
	
	\draw[->,line width = 0.7pt, double] (G)to(Gdec);
	
\end{tikzpicture}

%% file: tikzjournal/automateDeBeaudou.tex
\begin{tikzpicture}
	\node (s1) at (2,10) {
		\begin{tikzpicture}
			\node[noeud] (s) at (0,0) {};
			\draw (s) node[left,scale=2] {$s$};
			\node[noeud] (t) at (4,0) {};
			\draw (t) node[right,scale=2] {$t$};
			\draw[bend left,opacity=0.5] (s)to(t);
			\draw[bend right,opacity=0.5] (s)to(t);
			
			\draw[line width=0.25mm] (s) circle (0.25);
			
			\draw[rounded corners,line width=1mm] (1,1) rectangle (3,2);
			\draw[rounded corners,line width=1mm] (-1,-1) rectangle (5,1);
			\draw (2,1.5) node[scale=2.5] {$\mathbf{E_1}$};
		\end{tikzpicture}
	};
	
	\node (s2) at (10,10) {
		\begin{tikzpicture}
			\node[noeud] (s) at (0,0) {};
			\draw (s) node[left,scale=2] {$s$};
			\node[noeud] (t) at (4,0) {};
			\draw (t) node[right,scale=2] {$t$};
			\draw[bend left,opacity=0.5,postaction={-,draw,opacity=1,line width=0.25mm,dash pattern=on 10pt off 10cm}] (s)to node[noeud,pos=0.1,opacity=1](u){} (t);
			\draw[bend right,opacity=0.5] (s)to(t);
			
			\draw[line width=0.25mm] (u) circle (0.25);
			
			\draw[rounded corners,line width=1mm] (1,1) rectangle (3,2);
			\draw[rounded corners,line width=1mm] (-1,-1) rectangle (5,1);
			\draw (2,1.5) node[scale=2.5] {$\mathbf{E_2}$};
		\end{tikzpicture}
	};
	
	\node (A) at (6,5) {
		\begin{tikzpicture}
			\node[noeud] (s) at (0,0) {};
			\draw (s) node[left,scale=2] {$s$};
			\node[noeud] (t) at (4,0) {};
			\draw (t) node[right,scale=2] {$t$};
			\draw[bend left,opacity=0.5,postaction={-,draw,opacity=1,line width=0.25mm,dash pattern=on 45pt off 10cm}] (s)to node[pos=0.2,above,opacity=1]{$D$} node[noeud,pos=0.4,opacity=1](u){} (t);
			\draw[bend right,opacity=0.5,postaction={-,draw,opacity=1,line width=0.25mm,dash pattern=on 45pt off 10cm}] (s)to node[pos=0.2,below,opacity=1]{$D$} node[noeud,pos=0.4,opacity=1](w){} (t);
			
			\draw[line width=0.25mm] (u) circle (0.25);
			\draw (u) node[above=6 pt, scale=1] {$x$};
			
			\draw (w) node[below=5 pt, scale=1] {$y$};
			
			\draw[rounded corners,line width=1mm] (1,1.3) rectangle (3,2.3);
			\draw[rounded corners,line width=1mm] (-1,-1.3) rectangle (5,1.3);
			\draw (2,1.8) node[scale=2.5] {\textbf{A}};
		\end{tikzpicture}
	};
	
	\node (B) at (12,0) {
		\begin{tikzpicture}
			\node[noeud] (s) at (0,0) {};
			\draw (s) node[left,scale=2] {$s$};
			\node[noeud] (t) at (4,0) {};
			\draw (t) node[right,scale=2] {$t$};
			\draw[bend left,opacity=0.5,postaction={-,draw,opacity=1,line width=0.25mm,dash pattern=on 75pt off 10cm}] (s)to node[pos=0.2,above,opacity=1]{$D$} node[noeud,pos=0.4,opacity=1](u){} node[pos=0.55,below,opacity=1]{$D$} node[noeud,pos=0.7,opacity=1](v){} (t);
			\draw[bend right,opacity=0.5,postaction={-,draw,opacity=1,line width=0.25mm,dash pattern=on 45pt off 10cm}] (s)to node[pos=0.2,below,opacity=1]{$D$} node[noeud,pos=0.4,opacity=1](w){} (t);
			
			\draw[line width=0.25mm] (v) circle (0.25);
			
			\draw (u) node[above=5 pt, scale=1] {$x$};
			
			\draw (v) node[above=6 pt, scale=1] {$x'$};
			
			\draw (w) node[below=5 pt, scale=1] {$y$};
			
			\draw[rounded corners,line width=1mm] (1,1.3) rectangle (3,2.3);
			\draw[rounded corners,line width=1mm] (-1,-1.3) rectangle (5,1.3);
			\draw (2,1.8) node[scale=2.5] {\textbf{B}};
		\end{tikzpicture}
	};
	
	\node (C) at (0,0) {
		\begin{tikzpicture}
			\node[noeud] (s) at (0,0) {};
			\draw (s) node[left,scale=2] {$s$};
			\node[noeud] (t) at (4,0) {};
			\draw (t) node[right,scale=2] {$t$};
			\draw[bend left,opacity=0.5,postaction={-,draw,opacity=1,line width=0.25mm,dash pattern=on 75pt off 10cm}] (s)to node[pos=0.2,above,opacity=1]{$D$} node[noeud,pos=0.4,opacity=1](u){} node[pos=0.55,below,opacity=1]{$D$} node[noeud,pos=0.7,opacity=1](v){} (t);
			\draw[bend right,opacity=0.5,postaction={-,draw,opacity=1,line width=0.25mm,dash pattern=on 45pt off 10cm}] (s)to node[pos=0.2,below,opacity=1]{$D$} node[noeud,pos=0.4,opacity=1](w){} (t);
			
			\draw[line width=0.25mm] (w) circle (0.25);
			
			\draw (u) node[above=5 pt, scale=1] {$x$};
			
			\draw (v) node[above=5 pt, scale=1] {$x'$};
			
			\draw (w) node[below=6 pt, scale=1] {$y$};
			
			\draw[rounded corners,line width=1mm] (1,1.3) rectangle (3,2.3);
			\draw[rounded corners,line width=1mm] (-1,-1.3) rectangle (5,1.3);
			\draw (2,1.8) node[scale=2.5] {\textbf{C}};
		\end{tikzpicture}
	};
	
	\draw[->,line width = 0.7pt, double,bend left] (-1,12.5)to(s1.north);
	\draw[->,line width = 0.7pt, double] (s1)to(s2);
	\draw[->,line width = 0.7pt, double,out=270,in=90] (s2.south)to(A.north);
	\draw[->,line width = 0.7pt, double,bend left] (A.east)to(B.north);
	\draw[->,line width = 0.7pt, double] (B)to(C);
	
	\draw[->,line width = 0.7pt, double,bend left] (C.north)to node[midway,above left,scale=1.5]{update: $D \leftarrow 2D$} (A.west);
\end{tikzpicture}

%% file: tikzjournal/steps-strategy.tex
\begin{subfigure}{0.47\textwidth}
\centering

\scalebox{1}{
\begin{tikzpicture}
			\node[noeud,scale=.75] (s) at (0,0) {};
			\draw (s) node[left,scale=1.5] {$s$};
			\node[noeud,scale=.75] (t) at (4,0) {};
			\draw (t) node[right,scale=1.5] {$t$};
			\draw[bend left,opacity=0.5,postaction={-,draw,opacity=1,line width=0.25mm,dash pattern=on 55pt off 10cm}] (s)to node[pos=0.2,above,opacity=1]{$D$} node[noeud,pos=0.4,opacity=1,scale=.75](x){} node[noeud,pos=0.5,opacity=1, scale=.5](b1){} node[pos=0.55,below,opacity=1]{$e$} node[noeud,pos=0.6,opacity=1, scale=.5](b2){}   node[noeud,pos=0.75,opacity=1,scale=.75](v){}(t);
			\draw[bend right,opacity=0.5,postaction={-,draw,opacity=1,line width=0.25mm,dash pattern=on 60pt off 10cm}] (s)to node[pos=0.2,below,opacity=1]{$D$} node[noeud,pos=0.4,opacity=1,scale=.75](y){} node[noeud,pos=0.55,opacity=1,scale=.75](u){}  (t);
			
			\draw[line width=0.25mm] (u) circle (0.25);
			
			\draw[line width=0.25mm] (u)to (v);
			
			\draw[line width=0.5mm, color=red] (b1)to (b2);
			
			\draw (x) node[above=5 pt, scale=1] {$x$};

			\draw (u) node[below=6 pt, scale=1] {$u$};
			
			\draw (y) node[below=5 pt, scale=1] {$y$};
			\draw (v) node[below=5 pt, scale=1] {$v$};
			
		\end{tikzpicture}
		}
\caption{Step \ref{step: blocked edge} : blocked edge $e$ between states $\mathbf{A}$ and $\mathbf{B}$.}
\label{fig: blocked edge between A and B}
\end{subfigure} \hspace{0.04\textwidth}
\begin{subfigure}{0.47\textwidth}

\centering

\scalebox{1}{
\begin{tikzpicture}
			\node[noeud,scale=.75] (s) at (0,0) {};
			\draw (s) node[left,scale=1.5] {$s$};
			\node[noeud,scale=.75] (t) at (4,0) {};
			\draw (t) node[right,scale=1.5] {$t$};
			\draw[bend left,opacity=0.5,postaction={-,draw,opacity=1,line width=0.25mm,dash pattern=on 95pt off 10cm}] (s)to node[pos=0.2,above,opacity=1]{$D$} node[noeud,pos=0.4,opacity=1,scale=.75](x){} node[pos=0.55,below,opacity=1]{$D$} node[noeud,pos=0.7,opacity=1,scale=.75](xx){} node[noeud,pos=0.85,opacity=1,scale=.75](u){} (t);
			\draw[bend right,opacity=0.5,postaction={-,draw,opacity=1,line width=0.25mm,dash pattern=on 55pt off 10cm}] (s)to node[pos=0.2,below,opacity=1]{$D$} node[noeud,pos=0.4,opacity=1,scale=.75](y){}  node[noeud,pos=0.5,opacity=1, scale=.5](b1){} node[pos=0.55,below,opacity=1]{$e$} node[noeud,pos=0.6,opacity=1, scale=.5](b2){} node[noeud,pos=0.75,opacity=1,scale=.75](v){}(t);
			
			\draw[line width=0.25mm] (u) circle (0.25);
			
			\draw[line width=0.25mm] (u)to (v);
			
			\draw[line width=0.5mm, color=red] (b1)to (b2);
			
			\draw (x) node[above=5 pt, scale=1] {$x$};
			
			\draw (xx) node[above=5 pt, scale=1] {$x'$};
			
			\draw (u) node[above=6 pt, scale=1] {$u$};
			
			\draw (y) node[below=5 pt, scale=1] {$y$};
			\draw (v) node[below=5 pt, scale=1] {$v$};
			
		\end{tikzpicture}
		}
\caption{Step \ref{step: blocked edge} : blocked edge $e$ between states $\mathbf{C}$ and $\mathbf{A}$.}
\label{fig: blocked edge between C and A}
\end{subfigure}
\begin{subfigure}{0.47\textwidth}
\centering

\scalebox{1}{
\begin{tikzpicture}
			\node[noeud,scale=.75] (s) at (0,0) {};
			\draw (s) node[left,scale=1.5] {$s$};
			\node[noeud,scale=.75] (t) at (4,0) {};
			\draw (t) node[right,scale=1.5] {$t$};
			\draw[bend left,opacity=0.5,postaction={-,draw,opacity=1,line width=0.25mm,dash pattern=on 75pt off 10cm}] (s)to node[pos=0.2,above,opacity=1]{$D$} node[noeud,pos=0.4,opacity=1,scale=.75](x){} node[pos=0.55,below,opacity=1]{$\alpha D$}  node[noeud,pos=0.7,opacity=1,scale=.75](u){}(t);
			\draw[bend right,opacity=0.5,postaction={-,draw,opacity=1,line width=0.25mm,dash pattern=on 45pt off 10cm}] (s)to node[pos=0.2,below,opacity=1]{$D$} node[noeud,pos=0.4,opacity=1,scale=.75](y){} node[pos=0.55,below,opacity=1]{$?$} node[noeud,pos=0.7,opacity=1,scale=.75](v){}  (t);
			
			\draw[line width=0.25mm] (u) circle (0.25);
			
			\draw[line width=0.25mm] (u)to (v);

			\draw (x) node[above=5 pt, scale=1] {$x$};

			\draw (u) node[above=6 pt, scale=1] {$u$};
			
			\draw (y) node[below=5 pt, scale=1] {$y$};
			\draw (v) node[below=5 pt, scale=1] {$v$};
			
		\end{tikzpicture}
		}
\caption{Step \ref{step: open chord between A and B} : open vertical chord $uv$ between states $\mathbf{A}$ and $\mathbf{B}$.}
\label{fig: open chord between A and B}
\end{subfigure} \hspace{0.04\textwidth}
\begin{subfigure}{0.47\textwidth}

\centering

\scalebox{1}{
\begin{tikzpicture}
			\node[noeud,scale=.75] (s) at (0,0) {};
			\draw (s) node[left,scale=1.5] {$s$};
			\node[noeud,scale=.75] (t) at (4,0) {};
			\draw (t) node[right,scale=1.5] {$t$};
			\draw[bend left,opacity=0.5,postaction={-,draw,opacity=1,line width=0.25mm,dash pattern=on 75pt off 10cm}] (s)to node[pos=0.2,above,opacity=1]{$D$} node[noeud,pos=0.4,opacity=1,scale=.75](x){} node[pos=0.55,below,opacity=1]{$D$} node[noeud,pos=0.7,opacity=1,scale=.75](xx){} node[noeud,pos=0.85,opacity=1,scale=.75](v){} (t);
			\draw[bend right,opacity=0.5,postaction={-,draw,opacity=1,line width=0.25mm,dash pattern=on 70pt off 10cm}] (s)to node[pos=0.2,below,opacity=1]{$D$} node[noeud,pos=0.4,opacity=1,scale=.75](y){}   node[pos=0.52,above,opacity=1]{$\alpha D$} node[noeud,pos=0.65,opacity=1,scale=.75](u){}(t);
			
			\draw[line width=0.25mm] (u) circle (0.25);
			
			\draw[line width=0.25mm] (u)to (v);

			\draw (x) node[above=5 pt, scale=1] {$x$};
			
			\draw (xx) node[above=5 pt, scale=1] {$x'$};
			
			\draw (u) node[below=6 pt, scale=1] {$u$};
			
			\draw (y) node[below=5 pt, scale=1] {$y$};
			\draw (v) node[above=5 pt, scale=1] {$v$};
			
		\end{tikzpicture}
		}
\caption{Step \ref{step: open chord between C and A} : open vertical chord $uv$ between states $\mathbf{C}$ and $\mathbf{A}$.}
\label{fig: open chord between C and A}
\end{subfigure}

%% file: tikzjournal/unweighted-first.tex
\begin{tikzpicture}
	\node (g1) at (0,0) {
		\begin{tikzpicture}
			\node[noeud] (s1) at (0,1) {};
			\node[noeud] (u1t) at (1,2) {};
			\node[noeud] (u2t) at (2,2) {};
			\node[noeud] (u1b) at (1,0) {};
			\node[noeud] (u2b) at (2,0) {};
			\node[noeud] (t1) at (3,1) {};
			
			\draw[line width = 0.75mm, color = white!60!blue] \convexpath{s1,u1t,u2t,t1,u2t,u1t,s1}{0.31cm};
			\fill[white] \convexpath{s1,u1t,u2t,t1,u2t,u1t,s1}{0.30cm};
			\draw[line width = 0.75mm, color = white!30!malachite] \convexpath{s1,u1b,u2b,t1,u2b,u1b,s1}{0.23cm};
			\fill[white] \convexpath{s1,u1b,u2b,t1,u2b,u1b,s1}{0.22cm};
			
			\node[noeud] (s1) at (0,1) {};
			\node[noeud] (u1t) at (1,2) {};
			\node[noeud] (u2t) at (2,2) {};
			\node[noeud] (u1b) at (1,0) {};
			\node[noeud] (u2b) at (2,0) {};
			\node[noeud] (t1) at (3,1) {};
			
			\draw[line width = 1.1pt] (s1)to(u1t);
			\draw[line width = 1.1pt] (s1)to(u1b);
			\draw[line width = 1.1pt] (u1t)to(u2t);
			\draw[line width = 1.1pt] (u1b)to(u2b);
			\draw[line width = 1.1pt] (u2t)to(t1);
			\draw[line width = 1.1pt] (u2b)to(t1);
			\draw[line width = 1.1pt] (u1t)to(t1);
			\draw[line width = 1.1pt] (u1b)to(t1);
			
			\draw (s1) node[left,scale=2,xshift=-0.75mm] {$s$};
			\draw (t1) node[right,scale=2,xshift=0.75mm] {$t$};
			\draw (1.5,2.75) node[scale=1.5,color = blue] {\haut $S_1$};
			\draw (1.5,-0.75) node[scale=1.5, color = malachite] {\bas $S_2$};
		\end{tikzpicture}
	};
	
	\node (g2) at (12,0) {
		\begin{tikzpicture}
			\node[noeud] (s1) at (0,1) {};
			\node[noeud] (u1t) at (1,2) {};
			\node[noeud] (u2t) at (2,2) {};
			\node[noeud] (u1b) at (1,0) {};
			\node[noeud] (u2b) at (2,0) {};
			\node[noeud] (t1) at (3,1) {};
			
			\draw[free] (s1)to(u1t);
			\draw[free] (s1)to(u1b);
			\draw[free] (u1t)to(u2t);
			\draw (u1b)to(u2b);
			\draw (u2t)to(t1);
			\draw (u2b)to(t1);
			\draw[blocked] (u1t)to(t1);
			\draw (u1b)to(t1);
			
			\draw (s1) node[left,scale=2,xshift=-0.75mm] {$s$};
			\draw (t1) node[right,scale=2,xshift=0.75mm] {$t$};
			\draw[line width=0.75mm] (u1t) circle (0.25);
		\end{tikzpicture}
	};
	
	\node (g3) at (0,-6) {
		\begin{tikzpicture}
			\node[noeud] (s1) at (0,1) {};
			\node[noeud] (u1t) at (1,2) {};
			\node[noeud] (u2t) at (2,2) {};
			\node[noeud] (u1b) at (1,0) {};
			\node[noeud] (u2b) at (2,0) {};
			\node[cross out,draw=red,line width=0.75mm] (u2b) at (2,0) {};
			\node[noeud] (t1) at (3,1) {};
			
			\draw[free] (s1)to(u1t);
			\draw[free] (s1)to(u1b);
			\draw[free] (u1t)to(u2t);
			\draw[free] (u1b)to(u2b);
			\draw (u2t)to(t1);
			\draw (u2b)to(t1);
			\draw[blocked] (u1t)to(t1);
			\draw[free] (u1b)to(t1);
			
			\draw (s1) node[left,scale=2,xshift=-0.75mm] {$s$};
			\draw (t1) node[right,scale=2,xshift=0.75mm] {$t$};
			\draw[line width=0.75mm] (u1b) circle (0.25);
		\end{tikzpicture}
	};
	
	\node (g4) at (12,-6) {
		\begin{tikzpicture}
			\node[noeud] (s1) at (0,1) {};
			\node[noeud] (u1t) at (1,2) {};
			\node[noeud] (u2t) at (2,2) {};
			\node[noeud] (u1b) at (1,0) {};
			\node[noeud] (u2b) at (2,0) {};
			\node[cross out,draw=red,line width=0.75mm] (u2b) at (2,0) {};
			\node[noeud] (t1) at (3,1) {};
			
			\draw[free] (s1)to(u1t);
			\draw[free] (s1)to(u1b);
			\draw[free] (u1t)to(u2t);
			\draw[free] (u1b)to(u2b);
			\draw[free] (u2t)to(t1);
			\draw[free] (u2b)to(t1);
			\draw[blocked] (u1t)to(t1);
			\draw[free] (u1b)to(t1);
			
			\draw (s1) node[left,scale=2,xshift=-0.75mm] {$s$};
			\draw (t1) node[right,scale=2,xshift=0.75mm] {$t$};
			\draw[line width=0.75mm] (t1) circle (0.25);
		\end{tikzpicture}
	};
	
	\draw[->,line width = 0.7pt, double] (g1)to node[midway,text centered,text width=65mm,yshift=13mm]{Exploring the \haut with budget~1, gaining information on open and blocked edges. The red \hcorde is blocked, preventing us from reaching $t$.} (g2);
	\draw[->,line width = 0.7pt, double] (g2)to node[midway,sloped,text centered,text width=65mm,yshift=10mm]{Exploring the \bas with budget~2, we reveal an open \hcorde. We remove the now useless vertices (Step~\ref{step: reach horizontal chord}).} (g3);
	\draw[->,line width = 0.7pt, double] (g3)to node[midway,sloped,text centered,text width=65mm,yshift=8mm]{Target reached, with $T=4$ while the optimal offline path has length~2.} (g4);
\end{tikzpicture}

%% file: tikzjournal/unweighted-second.tex
\begin{tikzpicture}
	\node (g1) at (0,0) {
		\begin{tikzpicture}
			\node[noeud] (s3) at (0,1) {};
			\node[noeud] (v1t) at (1,2) {};
			\node[noeud] (v2t) at (2,2) {};
			\node[noeud] (v3t) at (3,2) {};
			\node[noeud] (v4t) at (4,2) {};
			\node[noeud] (v5t) at (5,2) {};
			\node[noeud] (v1b) at (1,0) {};
			\node[noeud] (v2b) at (2,0) {};
			\node[noeud] (v3b) at (3,0) {};
			\node[noeud] (v4b) at (4,0) {};
			\node[noeud] (v5b) at (5,0) {};
			\node[noeud] (t3) at (6,1) {};
			
			\draw[line width = 0.75mm, color = white!60!blue] \convexpath{s3,v1t,v2t,v3t,v4t,v5t,t3,v5t,v4t,v3t,v2t,v1t,s3}{0.31cm};
			\fill[white] \convexpath{s3,v1t,v2t,v3t,v4t,v5t,t3,v5t,v4t,v3t,v2t,v1t,s3}{0.30cm};
			\draw[line width = 0.75mm, color = white!30!malachite] \convexpath{s3,v1b,v2b,v3b,v4b,v5b,t3,v5b,v4b,v3b,v2b,v1b,s3}{0.23cm};
			\fill[white] \convexpath{s3,v1b,v2b,v3b,v4b,v5b,t3,v5b,v4b,v3b,v2b,v1b,s3}{0.22cm};
			
			\node[noeud] (s3) at (0,1) {};
			\node[noeud] (v1t) at (1,2) {};
			\node[noeud] (v2t) at (2,2) {};
			\node[noeud] (v3t) at (3,2) {};
			\node[noeud] (v4t) at (4,2) {};
			\node[noeud] (v5t) at (5,2) {};
			\node[noeud] (v1b) at (1,0) {};
			\node[noeud] (v2b) at (2,0) {};
			\node[noeud] (v3b) at (3,0) {};
			\node[noeud] (v4b) at (4,0) {};
			\node[noeud] (v5b) at (5,0) {};
			\node[noeud] (t3) at (6,1) {};
			
			\draw (s3)to(v1t);
			\draw (v1t)to(v2t);
			\draw (v2t)to(v3t);
			\draw (v3t)to(v4t);
			\draw (v4t)to(v5t);
			\draw (s3)to(v1b);
			\draw (v1b)to(v2b);
			\draw (v2b)to(v3b);
			\draw (v3b)to(v4b);
			\draw (v4b)to(v5b);
			\draw (v5t)to(t3);
			\draw (v5b)to(t3);
			\draw (v3t)to(v3b);
			\draw (v4t)to(v4b);
			\draw (v5t)to(v4b);
			\draw (v1t)to(v3b);
			
			\draw (s3) node[left,scale=2,xshift=-0.75mm] {$s_3$};
			\draw (t3) node[right,scale=2,xshift=0.75mm] {$t_3$};
			\draw (3,2.75) node[scale=1.5,blue] {\haut};
			\draw (3,-0.75) node[scale=1.5,malachite] {\bas};
		\end{tikzpicture}
	};
	
	\node (g2) at (14,0) {
		\begin{tikzpicture}
			\node[noeud] (s3) at (0,1) {};
			\node[noeud] (v1t) at (1,2) {};
			\node[noeud] (v2t) at (2,2) {};
			\node[noeud] (v3t) at (3,2) {};
			\node[noeud] (v4t) at (4,2) {};
			\node[noeud] (v5t) at (5,2) {};
			\node[noeud] (v1b) at (1,0) {};
			\node[noeud] (v2b) at (2,0) {};
			\node[noeud] (v3b) at (3,0) {};
			\node[noeud] (v4b) at (4,0) {};
			\node[noeud] (v5b) at (5,0) {};
			\node[noeud] (t3) at (6,1) {};
			
			\draw[free] (s3)to(v1t);
			\draw[free] (v1t)to(v2t);
			\draw[blocked] (v2t)to(v3t);
			\draw (v3t)to(v4t);
			\draw (v4t)to(v5t);
			\draw[free] (s3)to(v1b);
			\draw[free] (v1b)to(v2b);
			\draw[free] (v2b)to(v3b);
			\draw (v3b)to(v4b);
			\draw (v4b)to(v5b);
			\draw (v5t)to(t3);
			\draw (v5b)to(t3);
			\draw (v3t)to(v3b);
			\draw (v4t)to(v4b);
			\draw (v5t)to(v4b);
			\draw[blocked] (v1t)to(v3b);
			
			\draw (s3) node[left,scale=2,xshift=-0.75mm] {$s_3$};
			\draw (t3) node[right,scale=2,xshift=0.75mm] {$t_3$};
			\draw[line width=0.75mm] (v2t) circle (0.25);
		\end{tikzpicture}
	};
	
	\node (g3) at (0,-6) {
		\begin{tikzpicture}
			\node[noeud,opacity=0.5] (s3) at (0,1) {};
			\node[noeud,opacity=0.5] (v1t) at (1,2) {};
			\node[noeud,opacity=0.5] (v2t) at (2,2) {};
			\node[noeud] (v3t) at (3,2) {};
			\node[noeud] (v4t) at (4,2) {};
			\node[noeud] (v5t) at (5,2) {};
			\node[noeud,opacity=0.5] (v1b) at (1,0) {};
			\node[noeud,opacity=0.5] (v2b) at (2,0) {};
			\node[noeud] (v3b) at (3,0) {};
			\node[noeud] (v4b) at (4,0) {};
			\node[noeud] (v5b) at (5,0) {};
			\node[noeud] (t3) at (6,1) {};
			
			\draw[free,opacity=0.5] (s3)to(v1t);
			\draw[free,opacity=0.5] (v1t)to(v2t);
			\draw[blocked,opacity=0.5] (v2t)to(v3t);
			\draw (v3t)to(v4t);
			\draw (v4t)to(v5t);
			\draw[free,opacity=0.5] (s3)to(v1b);
			\draw[free,opacity=0.5] (v1b)to(v2b);
			\draw[free,opacity=0.5] (v2b)to(v3b);
			\draw (v3b)to(v4b);
			\draw (v4b)to(v5b);
			\draw (v5t)to(t3);
			\draw (v5b)to(t3);
			\draw (v3t)to(v3b);
			\draw (v4t)to(v4b);
			\draw (v5t)to(v4b);
			\draw[blocked,opacity=0.5] (v1t)to(v3b);
			
			\draw (s3) node[left,scale=2,xshift=-0.75mm,opacity=0.5] {$s_3$};
			\draw (t3) node[right,scale=2,xshift=0.75mm] {$t_3$};
			\draw[line width=0.75mm] (v3b) circle (0.25);
		\end{tikzpicture}
	};
	
	\node (g4) at (14,-6) {
		\begin{tikzpicture}
			\node[noeud] (v3t) at (3,2) {};
			\node[noeud] (v4t) at (4,2) {};
			\node[noeud] (v5t) at (5,2) {};
			\node[noeud] (v3b) at (3,0) {};
			\node[noeud] (v4b) at (4,0) {};
			\node[noeud] (v5b) at (5,0) {};
			\node[noeud] (t3) at (6,1) {};
			
			\draw[line width = 0.75mm, color = white!60!blue] \convexpath{v3b,v4b,v5b,t3,v5b,v4b,v3b}{0.31cm};
			\fill[white] \convexpath{v3b,v4b,v5b,t3,v5b,v4b,v3b}{0.30cm};
			\draw[line width = 0.75mm, color = white!30!malachite] \convexpath{v3b,v3t,v4t,v5t,t3,v5t,v4t,v3t,v3b}{0.23cm};
			\fill[white] \convexpath{v3b,v3t,v4t,v5t,t3,v5t,v4t,v3t,v3b}{0.22cm};
			
			\node[noeud,opacity=0.5] (s3) at (0,1) {};
			\node[noeud,opacity=0.5] (v1t) at (1,2) {};
			\node[noeud,opacity=0.5] (v2t) at (2,2) {};
			\node[noeud] (v3t) at (3,2) {};
			\node[noeud] (v4t) at (4,2) {};
			\node[noeud] (v5t) at (5,2) {};
			\node[noeud,opacity=0.5] (v1b) at (1,0) {};
			\node[noeud,opacity=0.5] (v2b) at (2,0) {};
			\node[noeud] (v3b) at (3,0) {};
			\node[noeud] (v4b) at (4,0) {};
			\node[noeud] (v5b) at (5,0) {};
			\node[noeud] (t3) at (6,1) {};
			
			\draw[free,opacity=0.5] (s3)to(v1t);
			\draw[free,opacity=0.5] (v1t)to(v2t);
			\draw[blocked,opacity=0.5] (v2t)to(v3t);
			\draw (v3t)to(v4t);
			\draw (v4t)to(v5t);
			\draw[free,opacity=0.5] (s3)to(v1b);
			\draw[free,opacity=0.5] (v1b)to(v2b);
			\draw[free,opacity=0.5] (v2b)to(v3b);
			\draw (v3b)to(v4b);
			\draw (v4b)to(v5b);
			\draw (v5t)to(t3);
			\draw (v5b)to(t3);
			\draw (v3t)to(v3b);
			\draw (v4t)to(v4b);
			\draw (v5t)to(v4b);
			\draw[blocked,opacity=0.5] (v1t)to(v3b);
			
			\draw (s3) node[left,scale=2,xshift=-0.75mm,opacity=0.5] {$s_3$};
			\draw (t3) node[right,scale=2,xshift=0.75mm] {$t_3$};
			\draw[line width=0.75mm] (v3b) circle (0.25);
			\draw (4.5,-0.75) node[scale=1.5,blue] {\haut};
			\draw (4.5,2.75) node[scale=1.5,malachite] {\bas};
		\end{tikzpicture}
	};

	\node (g5) at (0,-12) {
		\begin{tikzpicture}
			\node[noeud,opacity=0.5] (s3) at (0,1) {};
			\node[noeud,opacity=0.5] (v1t) at (1,2) {};
			\node[noeud,opacity=0.5] (v2t) at (2,2) {};
			\node[noeud] (v3t) at (3,2) {};
			\node[noeud] (v4t) at (4,2) {};
			\node[noeud] (v5t) at (5,2) {};
			\node[noeud,opacity=0.5] (v1b) at (1,0) {};
			\node[noeud,opacity=0.5] (v2b) at (2,0) {};
			\node[noeud] (v3b) at (3,0) {};
			\node[noeud] (v4b) at (4,0) {};
			\node[noeud] (v5b) at (5,0) {};
			\node[noeud] (t3) at (6,1) {};
			
			\draw[free,opacity=0.5] (s3)to(v1t);
			\draw[free,opacity=0.5] (v1t)to(v2t);
			\draw[blocked,opacity=0.5] (v2t)to(v3t);
			\draw (v3t)to(v4t);
			\draw (v4t)to(v5t);
			\draw[free,opacity=0.5] (s3)to(v1b);
			\draw[free,opacity=0.5] (v1b)to(v2b);
			\draw[free,opacity=0.5] (v2b)to(v3b);
			\draw[free] (v3b)to(v4b);
			\draw[free] (v4b)to(v5b);
			\draw (v5t)to(t3);
			\draw (v5b)to(t3);
			\draw[free] (v3t)to(v3b);
			\draw[free] (v4t)to(v4b);
			\draw[free] (v5t)to(v4b);
			\draw[blocked,opacity=0.5] (v1t)to(v3b);
			
			\draw (s3) node[left,scale=2,xshift=-0.75mm,opacity=0.5] {$s_3$};
			\draw (t3) node[right,scale=2,xshift=0.75mm] {$t_3$};
			\draw[line width=0.75mm] (v4b) circle (0.25);
		\end{tikzpicture}
	};
	
	\node (g6) at (14,-12) {
		\begin{tikzpicture}
			\node[noeud,opacity=0.5] (s3) at (0,1) {};
			\node[noeud,opacity=0.5] (v1t) at (1,2) {};
			\node[noeud,opacity=0.5] (v2t) at (2,2) {};
			\node[noeud] (v3t) at (3,2) {};
			\node[noeud] (v4t) at (4,2) {};
			\node[noeud] (v5t) at (5,2) {};
			\node[noeud,opacity=0.5] (v1b) at (1,0) {};
			\node[noeud,opacity=0.5] (v2b) at (2,0) {};
			\node[noeud] (v3b) at (3,0) {};
			\node[noeud] (v4b) at (4,0) {};
			\node[noeud] (v5b) at (5,0) {};
			\node[noeud] (t3) at (6,1) {};
			
			\draw[free,opacity=0.5] (s3)to(v1t);
			\draw[free,opacity=0.5] (v1t)to(v2t);
			\draw[blocked,opacity=0.5] (v2t)to(v3t);
			\draw[free] (v3t)to(v4t);
			\draw[free] (v4t)to(v5t);
			\draw[free,opacity=0.5] (s3)to(v1b);
			\draw[free,opacity=0.5] (v1b)to(v2b);
			\draw[free,opacity=0.5] (v2b)to(v3b);
			\draw[free] (v3b)to(v4b);
			\draw[free] (v4b)to(v5b);
			\draw[blocked] (v5t)to(t3);
			\draw (v5b)to(t3);
			\draw[free] (v3t)to(v3b);
			\draw[free] (v4t)to(v4b);
			\draw[free] (v5t)to(v4b);
			\draw[blocked,opacity=0.5] (v1t)to(v3b);
			
			\draw (s3) node[left,scale=2,xshift=-0.75mm,opacity=0.5] {$s_3$};
			\draw (t3) node[right,scale=2,xshift=0.75mm] {$t_3$};
			\draw[line width=0.75mm] (v4t) circle (0.25);
			\draw[->] (4.125,0.5) to (4.75,1.75) to (4.25,1.75);
		\end{tikzpicture}
	};
	
	\node (g7) at (0,-18) {
		\begin{tikzpicture}
			\node[noeud] (v5t) at (5,2) {};
			\node[noeud] (v4b) at (4,0) {};
			\node[noeud] (v5b) at (5,0) {};
			\node[noeud] (t3) at (6,1) {};
			
			\draw[line width = 0.75mm, color = white!60!blue] \convexpath{v4b,v5t,t3,v5t,v4b}{0.31cm};
			\fill[white] \convexpath{v4b,v5t,t3,v5t,v4b}{0.30cm};
			\draw[line width = 0.75mm, color = white!30!malachite] \convexpath{v4b,v5b,t3,v5b,v4b}{0.23cm};
			\fill[white] \convexpath{v4b,v5b,t3,v5b,v4b}{0.22cm};
			
			\node[noeud,opacity=0.5] (s3) at (0,1) {};
			\node[noeud,opacity=0.5] (v1t) at (1,2) {};
			\node[noeud,opacity=0.5] (v2t) at (2,2) {};
			\node[noeud,opacity=0.5] (v3t) at (3,2) {};
			\node[noeud,opacity=0.5] (v4t) at (4,2) {};
			\node[noeud] (v5t) at (5,2) {};
			\node[noeud,opacity=0.5] (v1b) at (1,0) {};
			\node[noeud,opacity=0.5] (v2b) at (2,0) {};
			\node[noeud,opacity=0.5] (v3b) at (3,0) {};
			\node[noeud] (v4b) at (4,0) {};
			\node[noeud] (v5b) at (5,0) {};
			\node[noeud] (t3) at (6,1) {};
			
			\draw[free,opacity=0.5] (s3)to(v1t);
			\draw[free,opacity=0.5] (v1t)to(v2t);
			\draw[blocked,opacity=0.5] (v2t)to(v3t);
			\draw[free,opacity=0.5] (v3t)to(v4t);
			\draw[free,opacity=0.5] (v4t)to(v5t);
			\draw[free,opacity=0.5] (s3)to(v1b);
			\draw[free,opacity=0.5] (v1b)to(v2b);
			\draw[free,opacity=0.5] (v2b)to(v3b);
			\draw[free,opacity=0.5] (v3b)to(v4b);
			\draw[free] (v4b)to(v5b);
			\draw[blocked] (v5t)to(t3);
			\draw (v5b)to(t3);
			\draw[free,opacity=0.5] (v3t)to(v3b);
			\draw[free,opacity=0.5] (v4t)to(v4b);
			\draw[free] (v5t)to(v4b);
			\draw[blocked,opacity=0.5] (v1t)to(v3b);
			
			\draw (s3) node[left,scale=2,xshift=-0.75mm,opacity=0.5] {$s_3$};
			\draw (t3) node[right,scale=2,xshift=0.75mm] {$t_3$};
			\draw[line width=0.75mm] (v4b) circle (0.25);
			\draw (5,2.75) node[scale=1.5,blue] {\haut};
			\draw (5,-0.75) node[scale=1.5,malachite] {\bas};
		\end{tikzpicture}
	};

	\node (g8) at (14,-18) {
		\begin{tikzpicture}
			\node[noeud,opacity=0.5] (s3) at (0,1) {};
			\node[noeud,opacity=0.5] (v1t) at (1,2) {};
			\node[noeud,opacity=0.5] (v2t) at (2,2) {};
			\node[noeud,opacity=0.5] (v3t) at (3,2) {};
			\node[noeud,opacity=0.5] (v4t) at (4,2) {};
			\node[noeud] (v5t) at (5,2) {};
			\node[noeud,opacity=0.5] (v1b) at (1,0) {};
			\node[noeud,opacity=0.5] (v2b) at (2,0) {};
			\node[noeud,opacity=0.5] (v3b) at (3,0) {};
			\node[noeud] (v4b) at (4,0) {};
			\node[noeud] (v5b) at (5,0) {};
			\node[noeud] (t3) at (6,1) {};
			
			\draw[free,opacity=0.5] (s3)to(v1t);
			\draw[free,opacity=0.5] (v1t)to(v2t);
			\draw[blocked,opacity=0.5] (v2t)to(v3t);
			\draw[opacity=0.5] (v3t)to(v4t);
			\draw[opacity=0.5] (v4t)to(v5t);
			\draw[free,opacity=0.5] (s3)to(v1b);
			\draw[free,opacity=0.5] (v1b)to(v2b);
			\draw[free,opacity=0.5] (v2b)to(v3b);
			\draw[free,opacity=0.5] (v3b)to(v4b);
			\draw[free] (v4b)to(v5b);
			\draw[blocked] (v5t)to(t3);
			\draw[free] (v5b)to(t3);
			\draw[free,opacity=0.5] (v3t)to(v3b);
			\draw[free,opacity=0.5] (v4t)to(v4b);
			\draw[free] (v5t)to(v4b);
			\draw[blocked,opacity=0.5] (v1t)to(v3b);
			
			\draw (s3) node[left,scale=2,xshift=-0.75mm,opacity=0.5] {$s_3$};
			\draw (t3) node[right,scale=2,xshift=0.75mm] {$t_3$};
			\draw[line width=0.75mm] (t3) circle (0.25);
		\end{tikzpicture}
	};
	
	\draw[->,line width = 0.7pt, double] (g1)to node[midway,text centered,text width=60mm,yshift=10mm]{Alternating between \haut and \bas, we reveal a closed edge which blocks this side (Step~\ref{step: blocked edge}).} (g2);
	\draw[->,line width = 0.7pt, double] (g2)to node[midway,sloped,text centered,text width=70mm,yshift=10mm]{Going to the other side until we reach a vertex incident with an open \vcorde, we can safely ignore everything behind.} (g3);
	\draw[->,line width = 0.7pt, double] (g3)to node[midway,text centered,text width=60mm,yshift=6mm]{We apply \expo on the remaining graph.} (g4);
	\draw[->,line width = 0.7pt, double] (g4)to node[midway,sloped,text centered,text width=65mm,yshift=10mm]{We reveal an open \vcorde between states $\mathbf{A}$ and $\mathbf{B}$: we check the distance on the opposite side (Step~\ref{step: open chord between A and B}).} (g5);
	\draw[->,line width = 0.7pt, double] (g5)to node[midway,text centered,text width=60mm,yshift=8mm]{We explore and do not see the last known vertex on the other side.} (g6);
	\draw[->,line width = 0.7pt, double] (g6)to node[midway,sloped,text centered,text width=70mm,yshift=8mm]{We remove the vertices behind the chord and apply \expo on the remaining graph.} (g7);
	\draw[->,line width = 0.7pt, double] (g7)to node[midway,text centered,text width=60mm,yshift=8mm]{After exploring more, we reveal a blocked edge and reach the target vertex on the other side.} node[midway,text centered,text width=60mm,yshift=-6mm]{We have $T=22$ while the optimal offline path has length~6.} (g8);
\end{tikzpicture}

%% file: tikzjournal/grapheDeDailly.tex
\begin{tikzpicture}
    \node[noeud] (s) at (0,1) {};
    \node[noeud] (t) at (5,1) {};
    \draw[out=60,in=120] (s)to node[noeud,scale=0.8,pos=0.2](x1){} node[noeud,scale=0.8,pos=0.4](x2){} node[noeud,scale=0.8,pos=0.6](x3){} node[noeud,scale=0.8,pos=0.8](x4){} (t);
    \draw[out=-60,in=240] (s)to node[noeud,scale=0.8,pos=0.2](y1){} node[noeud,scale=0.8,pos=0.4](y2){} node[noeud,scale=0.8,pos=0.6](y3){} node[noeud,scale=0.8,pos=0.8](y4){} (t);
    \foreach \I in {1,...,3} {
        \draw (x\I)to(t)to(y\I);
    }
    \foreach \I in {1,...,4} {
        \draw (x\I) node[above,scale=1.25] {$v_{\I}$};
        \pgfmathsetmacro{\J}{int(10-\I)}
        \draw (y\I) node[below,scale=1.25] {$v_{\J}$};
    }
    \draw (s) node[left,scale=1.5] {$s = v_0$};
    \draw (t) node[right,scale=1.5] {$v_5=t$};
\end{tikzpicture}

%% file: tikzjournal/grapheDeYanstar.tex
\begin{tikzpicture}
	\node[noeud] (s) at (0,2) {};
	\node[noeud] (t) at (8,2) {};
	
	\draw[bend left] (s) to node[above,pos=0.5]{$S$} (t);
	
	\draw[white,bend right] (s) to node[black,pos=0.09,sloped,yshift=1mm]{$H^{(1)}$} node[noeud,pos=0.2](y1){} node[black,pos=0.29,sloped,yshift=1mm]{$H^{(2)}$} node[noeud,pos=0.4](y2){} node[black,sloped,pos=0.55,scale=2]{$\ldots$} node[noeud,pos=0.7](y3){} node[black,pos=0.8,sloped,yshift=1mm]{$H^{(N)}$} node[noeud,pos=0.92](y4){} (t);
	
	\draw[bend left,dashed] (s) to (y1) to (y2);
	\draw[bend right,dashed] (s) to (y1) to (y2);
	\draw[bend left,dashed] (y3) to (y4);
	\draw[bend right,dashed] (y3) to (y4);
	\draw (y4) to node[below right]{$\varepsilon^*$} (t);
	
	\draw[out=60,in=180] (y1) to node[above,pos=0.25]{$\varepsilon^*$} (t);
	\draw[out=90,in=180] (y2) to node[above,pos=0.2]{$\varepsilon^*$} (t);
	\draw[out=150,in=180] (y3) to node[left,pos=0.25]{$\varepsilon^*$} (t);
	
	\draw (s) node[left,scale=1.25,xshift=-1mm] {$s=s_1$};
	\draw (t) node[right,scale=1.25,xshift=1mm] {$t$};
	\draw (y1) node[below left,yshift=-1mm,xshift=1mm] {$t_1=s_2$};
	\draw (y2) node[below,yshift=-1mm] {$t_2=s_3$};
	\draw (y3) node[below right,yshift=-1mm,xshift=-6mm] {$t_{N-1}=s_N$};
	\draw (y4) node[below right,yshift=-1mm,xshift=-2mm] {$t_N$};
\end{tikzpicture}

%% file: tikzjournal/exampleH2star.tex
\begin{tikzpicture}
	\node[noeud] (s) at (0,2) {};
	\node[noeud] (t) at (12,2) {};
	
	\draw[out=60,in=120] (s) to node[above,pos=0.5]{$S_2 = 6$} (t);
	
	\draw[white] (s) to node[noeud,pos=0.15](y1){} node[noeud,pos=0.3](y2){} 
	node[noeud,pos=0.45](y3){} 
	node[noeud,pos=0.6](y4){}
	node[noeud,pos=0.75](y5){}
	node[noeud,pos=0.9](y6){} (t);
	
\draw[out=50,in=130] (s) to node[above,pos=0.5]{$2$} (y1);
\foreach \i in {1,...,4}
{
    \pgfmathtruncatemacro{\j}{\i + 1}
	\draw[out=50,in=130] (y\i) to node[above,pos=0.6,yshift=-0.5mm]{$2$} (y\j);
}
\draw[out=50,in=130] (y5) to node[above,pos=0.8,yshift=-0.5mm]{$2$} (y6);
    
\draw[bend right] (s) to node[pos=0.25,xshift=-1mm,yshift=-2mm]{1} node[noeud,pos=0.5](y1-b){} node[pos=0.75,xshift=1mm,yshift=-2mm]{$\varepsilon^*$} (y1);
\foreach \i in {2,...,6}
{
    \pgfmathtruncatemacro{\j}{\i - 1}
    \draw[bend right] (y\j) to node[pos=0.25,xshift=-1mm,yshift=-2mm]{1} node[noeud,pos=0.5](y\i-b){} node[pos=0.75,xshift=1mm,yshift=-2mm]{$\varepsilon^*$} (y\i);
}
	\draw (y6) to node[below]{$\varepsilon^*$} (t);
	
	\draw[out=70,in=130] (y1) to node[above,pos=0.2]{$\varepsilon^*$} (t);
	\draw[out=65,in=132] (y2) to node[above,pos=0.2]{$\varepsilon^*$} (t);
	\draw[out=60,in=134] (y3) to node[left,pos=0.25]{$\varepsilon^*$} (t);
	\draw[out=57,in=136] (y4) to node[above,pos=0.25]{$\varepsilon^*$} (t);
	\draw[out=57,in=138] (y5) to node[left,pos=0.3,yshift=1mm]{$\varepsilon^*$} (t);
	
 	\draw (s) node[left,scale=1.25,xshift=-1mm] {$s$};
 	\draw (t) node[right,scale=1.25,xshift=1mm] {$t$};
\end{tikzpicture}

%% file: tikzjournal/trajetYanstar.tex
\begin{tikzpicture}
	\node[noeud] (s) at (0,2) {};
	\node (s1) at (0,0.9) {};
	\node (s2) at (0,0.5) {};
	\node (s3) at (0,3.1) {};
	\node[noeud] (t) at (8,2) {};
	\node (t1) at (7.5,0.9) {};
	\node (t2) at (7.5,0.5) {};
	\node (t3) at (8,3.1) {};
	
	\draw[out=40,in=140] (s) to node[above,pos=0.5]{$S_i$} (t);
	
	\draw[white,bend right] (s) to node[black,pos=0.09,sloped,yshift=1mm]{$H_{i-1}^{(1)}$} node[noeud,pos=0.2](y1){} node[black,pos=0.31,sloped,yshift=1mm]{$H_{i-1}^{(2)}$} node[noeud,pos=0.43](y2){} node[black,sloped,pos=0.56,scale=2]{$\ldots$} node[noeud,pos=0.7](y3){} node[black,pos=0.81,sloped,yshift=1mm]{$H_{i-1}^{(N_i)}$} node[noeud,pos=0.92](y4){} (t);
	
	\draw[out=20,in=110,dashed] (s) to (y1) to (y2);
	\draw[out=-60,in=-160,dashed] (s) to (y1) to (y2);
	\draw[out=60,in=160,dashed] (y3) to (y4);
	\draw[out=-10,in=-110,dashed] (y3) to (y4);
	\draw[red] (y4) to node[below right]{$\varepsilon^*$} (t);
	
	\draw[red,out=60,in=180] (y1) to node[above,pos=0.25]{$\varepsilon^*$} (t);
	\draw[red,out=90,in=180] (y2) to node[above,pos=0.2]{$\varepsilon^*$} (t);
	\draw[red,out=150,in=180] (y3) to node[left,pos=0.25]{$\varepsilon^*$} (t);
	
	\draw (s) node[left,scale=1.25,xshift=-1mm] {$s$};
	\draw (t) node[right,scale=1.25,xshift=1mm] {$t$};
	\draw (y4) node[below right,yshift=-1mm,xshift=-2mm] {$t_N$};
	
	\draw[blue,bend right,line width = 2pt,->] (s1) to node[above,pos=0.5,yshift=1mm]{outward: $\ge N_iS_{i-1}$} (t1);
	\draw[orange,bend left,line width = 2pt,->] (t2) to node[below,pos=0.5]{return: $\ge d$} (s2);
	\draw[black!10!green,bend left,line width = 2pt,->] (s3) to node[above,pos=0.5]{last traversal: $= S_i$} (t3);
\end{tikzpicture}

%% file: arxiv_journal_version.bbl
\begin{thebibliography}{10}

\bibitem{AkSaAr16}
V.~Aksakalli, O.~F. Sahin, and I.~Ari.
\newblock An {AO\({}^{\mbox{*}}\)} based exact algorithm for the canadian
  traveler problem.
\newblock {\em {INFORMS} Journal on Computing}, 28(1):96--111, 2016.

\bibitem{AlYiAk21}
A.~F. Alkaya, S.~Yildirim, and V.~Aksakalli.
\newblock Heuristics for the {Canadian Traveler Problem} with neutralizations.
\newblock {\em Comput. Ind. Eng.}, 159:107488, 2021.

\bibitem{BaCuRa93}
R.~A. Baeza{-}Yates, J.~C. Culberson, and G.~J.~E. Rawlins.
\newblock Searching in the plane.
\newblock {\em Inf. Comput.}, 106(2):234--252, 1993.

\bibitem{Ba94}
B.~S. Baker.
\newblock Approximation algorithms for {NP}-complete problems on planar graphs.
\newblock {\em J. {ACM}}, 41(1):153--180, 1994.

\bibitem{BaSc91}
A.~Bar{-}Noy and B.~Schieber.
\newblock The {C}anadian {T}raveller {P}roblem.
\newblock In {\em Proc. of ACM/SIAM SODA}, pages 261--270, 1991.

\bibitem{BeBe24}
L.~Beaudou, P.~Berg{\'{e}}, V.~Chernyshev, A.~Dailly, Y.~Gerard, A.~Lagoutte,
  V.~Limouzy, and L.~Pastor.
\newblock The {Canadian Traveller Problem} on outerplanar graphs.
\newblock In {\em Procs. of {MFCS}}, volume 306 of {\em LIPIcs}, pages
  19:1--19:16, 2024.

\bibitem{BeNe70}
A.~Beck and D.~J. Newman.
\newblock Yet more on the linear search problem.
\newblock {\em Isr. J. Math.}, 8(4):419--429, 1970.

\bibitem{BeBa23}
J.~Becker and R.~Batta.
\newblock Canadian prize collection problem.
\newblock {\em Military Operations Research}, 28(2):55--92, 2023.

\bibitem{Be63}
R.~Bellman.
\newblock Problem 63-9, an optimal search.
\newblock {\em SIAM review}, 5(3):274--274, 1963.

\bibitem{BeWe15}
M.~Bender and S.~Westphal.
\newblock An optimal randomized online algorithm for the k-{Canadian}
  {Traveller} {Problem} on node-disjoint paths.
\newblock {\em J. Comb. Optim.}, 30(1):87--96, 2015.

\bibitem{BeDeGuLe19}
P.~Berg{\'{e}}, J.~Desmarchelier, W.~Guo, A.~Lefebvre, A.~Rimmel, and
  J.~Tomasik.
\newblock Multiple canadians on the road: minimizing the distance competitive
  ratio.
\newblock {\em J. Comb. Optim.}, 38(4):1086--1100, 2019.

\bibitem{BeSa23}
P.~Berg{\'{e}} and L.~Sala{\"{u}}n.
\newblock The influence of maximum (\emph{s}, \emph{t})-cuts on the
  competitiveness of deterministic strategies for the {Canadian Traveller
  Problem}.
\newblock {\em Theor. Comput. Sci.}, 941:221--240, 2023.

\bibitem{BnFeSh09}
Z.~Bnaya, A.~Felner, and S.~E. Shimony.
\newblock Canadian traveler problem with remote sensing.
\newblock In {\em Proc. of {IJCAI}}, pages 437--442, 2009.

\bibitem{BoEl98}
A.~Borodin and R.~El{-}Yaniv.
\newblock {\em Online computation and competitive analysis}.
\newblock Cambridge Univ. Press, 1998.

\bibitem{ChChWuWu15}
H.~Chan, J.~Chang, H.~Wu, and T.~Wu.
\newblock {The $k$-Canadian Traveller Problem on Equal-Weight Graphs}.
\newblock {\em Proc. of WCMCT}, pages 135--137, 2015.

\bibitem{CoGo96}
R.~M. Corless, G.~H. Gonnet, D.~E.~G. Hare, D.~J. Jeffrey, and D.~E. Knuth.
\newblock On the lambert {W} function.
\newblock {\em Advances in Computational mathematics}, 5:329--359, 1996.

\bibitem{DeHuLiSa14}
E.~D. Demaine, Y.~Huang, C.~Liao, and K.~Sadakane.
\newblock Canadians {Should} {Travel} {Randomly}.
\newblock {\em Proc. of ICALP}, pages 380--391, 2014.

\bibitem{Di12}
R.~Diestel.
\newblock {\em Graph Theory, 4th Edition}, volume 173 of {\em Graduate texts in
  mathematics}.
\newblock Springer, 2012.

\bibitem{EyKeHe10}
P.~Eyerich, T.~Keller, and M.~Helmert.
\newblock High-quality policies for the {Canadian Traveler's Problem}.
\newblock In {\em Proces. of {AAAI}}, pages 51--58. {AAAI} Press, 2010.

\bibitem{farkas1902theorie}
J.~Farkas.
\newblock Theorie der einfachen ungleichungen.
\newblock {\em Journal f{\"u}r die reine und angewandte Mathematik (Crelles
  Journal)}, 1902(124):1--27, 1902.

\bibitem{fried2013complexity}
D.~Fried, S.~E. Shimony, A.~Benbassat, and C.~Wenner.
\newblock Complexity of canadian traveler problem variants.
\newblock {\em Theor. Comput. Sci.}, 487:1--16, 2013.

\bibitem{Gartner}
Ji\v{r}\'{\i} G\"{a}rtner, Bernd;~Matou\v{s}ek.
\newblock {\em Understanding and Using Linear Programming.}, page 81–104.
\newblock Berlin: Springer., 2006.

\bibitem{HaXe23}
N.~Hahn and M.~Xefteris.
\newblock The covering canadian traveller problem revisited.
\newblock In {\em Procs. of {MFCS}}, volume 272 of {\em LIPIcs}, pages
  53:1--53:12, 2023.

\bibitem{KaReTa96}
M.-Y. Kao, J.~H. Reif, and S.~T. Tate.
\newblock Searching in an unknown environment: An optimal randomized algorithm
  for the cow-path problem.
\newblock {\em Information and computation}, 131(1):63--79, 1996.

\bibitem{LiHu14}
C.~Liao and Y.~Huang.
\newblock The covering canadian traveller problem.
\newblock {\em Theor. Comput. Sci.}, 530:80--88, 2014.

\bibitem{LiScTh01}
L.~V. Lita, J.~Schulte, and S.~Thrun.
\newblock A system for multi-agent coordination in uncertain environments.
\newblock In {\em Procs. of {AGENTS}}, pages 21--22. {ACM}, 2001.

\bibitem{PaYa91}
C.~Papadimitriou and M.~Yannakakis.
\newblock Shortest paths without a map.
\newblock {\em Theor. Comput. Sci.}, 84(1):127--150, 1991.

\bibitem{ShSa17}
D.~Shiri and F.~S. Salman.
\newblock {On the online multi-agent {$O-D$} $k$-Canadian Traveler Problem}.
\newblock {\em J. Comb. Optim.}, 34(2):453--461, 2017.

\bibitem{We08}
S.~Westphal.
\newblock A note on the $k$-{Canadian} {Traveller} {Problem}.
\newblock {\em Inform. Proces. Lett.}, 106(3):87--89, 2008.

\bibitem{XuHuSuZh09}
Y.~Xu, M.~Hu, B.~Su, B.~Zhu, and Z.~Zhu.
\newblock The {C}anadian traveller problem and its competitive analysis.
\newblock {\em J. Comb. Optim.}, 18(2):195--205, 2009.

\end{thebibliography}
